%% file: main.tex
\documentclass[runningheads]{llncs}

\input{packages}

\input{macros}
\begin{document}

%% Title information
\title{Predictive Monitoring with Strong Trace Prefixes}    

\author{}
\institute{}

\author{Zhendong Ang \and
Umang Mathur}

\authorrunning{Z. Ang and U. Mathur}

% \email{zhendong.ang@u.nus.edu} \and
% \email{umathur@comp.nus.edu.sg}     
\institute{National University of Singapore, Singapore}

\maketitle

% \vspace{-0.5in}
\input{abstract}

% \keywords{concurrency, runtime verification, dynamic analysis, trace languages}

% \setlength{\belowcaptionskip}{-5pt}

\input{intro}

\input{prelim}
\input{prefix-def}
\input{complexity}
\input{readsfrom}
\input{conn-syncp}
\input{exp}
\input{conclusion}

\newpage

\bibliographystyle{splncs04}
\bibliography{refs}

\newpage
\appendix
\input{app-proofs-strong-prefixes}
\input{app-proofs-complexity}
\input{app-proofs-rf}
\input{app-proofs-syncp}
\input{app-exp}

\end{document}

%% file: packages.tex
\usepackage[T1]{fontenc}
\usepackage[utf8]{inputenc}
\usepackage{amsmath}
\usepackage{amssymb}
\usepackage{stmaryrd}
\usepackage{adjustbox}
\usepackage{multirow}
\usepackage{mathtools}
\usepackage[inline]{enumitem}
\usepackage{xspace}
\usepackage{comment}
\usepackage{xcolor}
\usepackage[ruled,vlined,resetcount,linesnumbered]{algorithm2e}
\usepackage{listings}
\usepackage[caption=false]{subfig}
\usepackage{tikz}
\usetikzlibrary{automata, positioning, arrows, shapes,shapes.geometric, calc}

\usepackage{float}

\usepackage{wrapfig}
\usepackage{tabularx}
\usepackage[normalem]{ulem} %for wiggle in code
\usepackage{contour}

\lstdefinestyle{mystyle}{
  basicstyle=\ttfamily,
  columns=fullflexible,
  keepspaces=true,
  keywordstyle=\color{green},
}
\usepackage{breqn}
\usepackage[htt]{hyphenat}

\usepackage{thmtools}

%% file: macros.tex
\input{exe_macro}

\newenvironment{myparagraph}[1]{\vspace*{0.05in}\noindent{\bf #1.}}{}

\makeatletter
\DeclareFontFamily{OMX}{MnSymbolE}{}
\DeclareSymbolFont{MnLargeSymbols}{OMX}{MnSymbolE}{m}{n}
\SetSymbolFont{MnLargeSymbols}{bold}{OMX}{MnSymbolE}{b}{n}
\DeclareFontShape{OMX}{MnSymbolE}{m}{n}{
    <-6>  MnSymbolE5
   <6-7>  MnSymbolE6
   <7-8>  MnSymbolE7
   <8-9>  MnSymbolE8
   <9-10> MnSymbolE9
  <10-12> MnSymbolE10
  <12->   MnSymbolE12
}{}
\DeclareFontShape{OMX}{MnSymbolE}{b}{n}{
    <-6>  MnSymbolE-Bold5
   <6-7>  MnSymbolE-Bold6
   <7-8>  MnSymbolE-Bold7
   <8-9>  MnSymbolE-Bold8
   <9-10> MnSymbolE-Bold9
  <10-12> MnSymbolE-Bold10
  <12->   MnSymbolE-Bold12
}{}

\let\llangle\@undefined
\let\rrangle\@undefined
\DeclareMathDelimiter{\llangle}{\mathopen}%
                     {MnLargeSymbols}{'164}{MnLargeSymbols}{'164}
\DeclareMathDelimiter{\rrangle}{\mathclose}%
                     {MnLargeSymbols}{'171}{MnLargeSymbols}{'171}
\makeatother

\newcommand{\figlabel}[1]{\label{fig:#1}}

\newcommand{\figref}[1]{Figure~\ref{fig:#1}}
\newcommand{\seclabel}[1]{\label{sec:#1}}
\newcommand{\secref}[1]{Section~\ref{sec:#1}}
\newcommand{\exlabel}[1]{\label{ex:#1}}
\newcommand{\exref}[1]{Example~\ref{ex:#1}}
\newcommand{\deflabel}[1]{\label{def:#1}}
\newcommand{\defref}[1]{Definition~\ref{def:#1}}

\newcommand{\thmlabel}[1]{\label{thm:#1}}
\newcommand{\thmref}[1]{Theorem~\ref{thm:#1}}

\newcommand{\lemlabel}[1]{\label{lem:#1}}
\newcommand{\lemref}[1]{Lemma~\ref{lem:#1}}

\newcommand{\equlabel}[1]{\label{eq:#1}}
\newcommand{\equref}[1]{Equation~(\ref{eq:#1})}
\newcommand{\applabel}[1]{\label{app:#1}}
\newcommand{\appref}[1]{Appendix~\ref{app:#1}}

\newcommand{\tablabel}[1]{\label{tab:#1}}
\newcommand{\tabref}[1]{Table~\ref{tab:#1}}

\newcommand{\set}[1]{\{#1\}}
\newcommand{\setpred}[2]{\set{#1 \,|\, #2}}

\newcommand{\proj}[2]{#1|_{#2}}

\renewcommand{\emptyset}{\varnothing}

\newcommand{\tr}{\sigma}

\newcommand{\ev}[2]{{#2}^{#1}}

\newcommand{\creorder}[1]{\textsf{CReorderings}(#1)}
\newcommand{\events}[1]{\textsf{Events}_{#1}}
\newcommand{\trord}[1]{<_{#1}}
\newcommand{\rf}[2]{\textsf{rf}_{#1}(#2)}

\newcommand{\alphabet}{\Sigma}
\newcommand{\rwl}{\mathsf{RWL}}
\newcommand{\wf}{\mathsf{WF}}
\newcommand{\lbl}[1]{\mathsf{lab}(#1)}

\newcommand{\threads}{\mathcal{T}}

\newcommand{\vars}{\mathcal{X}}
\newcommand{\locks}{\mathcal{L}}
\newcommand{\lk}{\ell}
\newcommand{\rd}{\mathtt{r}}
\newcommand{\wt}{\mathtt{w}}
\newcommand{\acq}{\mathtt{acq}}
\newcommand{\rel}{\mathtt{rel}}

\newcommand{\op}{\mathsf{op}}

\newcommand{\dep}{\mathbb{D}}

\newcommand{\wdth}{\alpha}
\newcommand{\mazeq}[1]{\sim_{#1}}

\newcommand{\eqcl}[2]{\llbracket#2\rrbracket_{#1}}

\newcommand{\dstrong}{\mathbb{S}}
\newcommand{\dweak}{\mathbb{W}} 
\newcommand{\strongqo}[2]{\preccurlyeq^{#2}_{#1}}

\newcommand{\idealqo}[1]{\sqsubseteq_{#1}}
\newcommand{\nidealqo}[1]{\not\!\sqsubseteq_{#1}}
\newcommand{\strongdcl}[3]{\langle\mkern-5mu\langle{#1}|\mkern-3mu|^{#3}_{#2}}
\newcommand{\idealdcl}[2]{\llbracket{#1}|\mkern-3mu|_{#2}}
\newcommand{\rfqo}{\trianglelefteq_{\mathsf{rf}}}
\newcommand{\rfdcl}[1]{\langle\mkern-5mu\langle{#1}|\mkern-3mu|_{\mathsf{rf}}}
\newcommand{\syncp}[1]{\texttt{SyncP}(#1)}
\newcommand{\confp}[1]{\texttt{ConfP}(#1)}

\definecolor{eventColor}{rgb}{0.86, 0.86, 0.95} % light blue
\definecolor{babyblueeyes}{rgb}{0.63, 0.79, 0.95}
\definecolor{lightskyblue}{rgb}{0.7, 0.81, 0.98}
\definecolor{lavender}{rgb}{0.9, 0.9, 0.98}
\definecolor{classColor}{rgb}{0.13, 0.53, 0.69}
\definecolor{methodColor}{rgb}{0.01, 0.16, 0.48}

\contourlength{0.3pt}

\newcommand{\patt}[2]{\alphabet^*{#1}_1\alphabet^*\dots\alphabet^*{#1}_{#2}\alphabet^*}

\newcommand{\seq}[2]{{#1}_1, \dots, {#1}_{#2}}

\newcommand{\tsan}{\textsc{ThreadSanitizer}\xspace}

\newcommand{\ucomment}[1]{\textcolor{red}{#1}}
\newcommand{\zacomment}[1]{}

%% file: exe_macro.tex
\newcommand{\execution}[2]{
\scalebox{0.9}{
  \begin{tikzpicture}%
    \foreach \x in {1,...,#1}
    \node[right] at (1.5*\x+0.2,0.25) {$t_{\x}$};
    \draw (1.2,0) -- (#1*1.5+1.2,0);%
    \pgfmathsetmacro{\y}{1};%
    #2%
    \draw (1.2,0) -- (1.2,-0.4*\y);%
    \draw (#1*1.5+1.2,0) -- (#1*1.5+1.2,-0.4*\y);%
    \foreach \x in {2,...,#1}
    \draw[dashed] (1.5*\x-0.3,0) -- (1.5*\x-0.3,-0.4*\y);%
    \draw (1.2,-0.4*\y) -- (#1*1.5+1.2,-0.4*\y);%
  \end{tikzpicture}
}
}

\newcommand{\smallexecution}[2]{
\scalebox{0.7}{
  \begin{tikzpicture}%
    \foreach \x in {1,...,#1}
    \node[right] at (\x,0.25) {$t_{\x}$};
    \draw (1.2-0.3,0) -- (#1*1.2,0);%
    \pgfmathsetmacro{\y}{1};%
    #2%
    \foreach \x in {2,...,#1}
    \draw[dashed] (1*\x-0.3,0) -- (1*\x-0.3,-0.4*\y);%
    \draw (1.2,-0.4*\y) -- (#1*1.2,-0.4*\y);%
  \end{tikzpicture}
}
}

\newcommand{\figev}[2]{
\pgfmathsetmacro{\y}{\y+1};
\pgfmathsetmacro{\y}{\y-1};
\node [left] at (1.25,-0.4*\y)  {\pgfmathprintnumber{\y}};%
\node at (#1*1.5 + 0.45,-0.4*\y) { #2 };%
\pgfmathsetmacro{\y}{\y+1};
}
\newcommand{\smallfigev}[2]{
\pgfmathsetmacro{\y}{\y+1};
\pgfmathsetmacro{\y}{\y-1};
\node [left] at (1,-0.4*\y)  {\pgfmathprintnumber{\y}};%
\node at (#1+0.3,-0.4*\y) { #2 };%
\pgfmathsetmacro{\y}{\y+1};
}

\newcommand{\executionfull}[9]{
\scalebox{#3}{
  \begin{tikzpicture}
    \foreach \x in {1,...,#1}
    \node[right] at (#4*\x+#6, #8) {$t_{\x}$}; % Print tids
    \draw (#7,0) -- (#1*#4+#7,0); % Draw top boundary
    \pgfmathsetmacro{\y}{1};
    #2 % Draw the execution
    \draw (#7,0) -- (#7,-#5*\y); % Draw left boundary
    \draw (#1*#4+#7,0) -- (#1*#4+#7,-#5*\y); % Draw right boundary
    \draw (#7,-#5*\y) -- (#1*#4+#7,-#5*\y); % Draw  bottom boundary
    \ifthenelse{#9 = 1}{
      \foreach \x in {2,...,#1}
      \draw[dashed] (#4*\x+#7-#4,0) -- (#4*\x+#7-#4,-#5*\y); % Draw inter-thread boundaries
    }{}
  \end{tikzpicture}
}
}

\newcommand{\figevfull}[9]{
\ifthenelse{#7 = 1}{
  \ifthenelse{#8 = -1}{
    \node [left] at (#5,(-#4*\y))  {\pgfmathprintnumber{\y}};%
  }{
    \node [left] at (#5,(-#4*\y))  {#9};%
  }
}{}
\node at (#1*#3 + #6,(-#4*\y)) {$ #2 $};%
\pgfmathsetmacro{\y}{\y+1};
}

%% file: abstract.tex
%!TEX root=./main.tex

\begin{abstract}
Runtime predictive analyses enhance coverage of traditional dynamic
analyses based bug detection techniques by identifying a space of
feasible reorderings of the observed execution and determining if any of these witnesses
the violation of some desired safety property.
% The space of feasible reorderings considered by the analysis is key in determining its
% running time as well as its exhaustiveness (or \emph{predictive power}).
% Reasoning based on Mazurkiewicz traces (or simply \emph{traces}) --- where the space of reorderings is
% those obtained by repeatedly commuting neighboring
% independent actions --- has emerged to be a simple yet means to model the class of reorderings,
% and has been key to practical performance of predicive techniques.
The most popular approach for modelling the space of feasible reorderings 
is through Mazurkiewicz's trace equivalence.
The simplicity of the framework also gives rise to efficient predictive
analyses, and has been the de facto means for obtaining space and time efficient 
algorithms for monitoring concurrent programs.

In this work, we investigate how to enhance the predictive power 
of trace-based reasoning, while still retaining the algorithmic benefits it offers.
Towards this, we extend trace theory by naturally embedding a class of prefixes, 
which we call \emph{strong trace prefixes}. 
We formally characterize strong trace prefixes using an enhanced dependence relation,
study its predictive power and establish a tight connection to the previously proposed
notion of synchronization preserving correct reorderings
developed in the context of data race and deadlock prediction.
We then show that despite the enhanced predictive power,
strong trace prefixes continue to enjoy the algorithmic
benefits of Mazurkiewicz traces in the context of prediction against
co-safety properties, and derive new algorithms for synchronization preserving data races
and deadlocks with better asymptotic space and time usage.
We also show that strong  trace prefixes can capture more violations
of pattern languages.
We implement our proposed algorithms 
% for predictive monitoring using prefix traces,
and our evaluation confirms the practical utility of reasoning based on strong prefix traces.
\end{abstract}

%% file: intro.tex
%!TEX root=./main.tex

\section{Introduction}
\seclabel{intro}

Dynamic analysis has emerged as a popular class of techniques for 
ensuring reliability of large scale software, owing to their 
scalability and soundness (no false positives).
At a high level, such techniques solve the membership problem ---
given an execution $\tr$, typically modelled as a sequence of events, 
does $\tr$ belong to $L_{\text{bug}}$, a chosen set of executions 
that exhibit some undesired behaviour.
In the context of concurrent software, however, such a naive testing paradigm
suffers from poor coverage, since
even under the ideal input, the execution $\tr$ observed at the time of testing,
may not reveal the presence of bug (membership in $L_{\text{bug}}$),
because of the non-determinism due to thread scheduling.
Runtime prediction, which is also the subject this work,
has emerged as a systematic approach to enhance
vanilla dynamic analyses~\cite{wcp2017,RVPredict2014,tuncc2023sound}.
Instead of solving the vanilla membership problem ($\tr \in L_{\textsf{bug}}$),
runtime predictive techniques solve the \emph{predictive membership} or
\emph{predictive monitoring} problem ---
they generalize the 
observed execution $\tr$ to a larger set of executions $S_\tr$ and check if 
there is some execution in $S_\tr$ that belongs to $L_{\text{bug}}$.
% To guarantee soundness, the analysis must ensure that $S_\tr$ only consists of 
% feasible executions, no matter
% what underlying program $\tr$ was observed from.
% \ucomment{Give an example.}

The predictive power (how often real bugs are identified) as well as the
speed of a runtime predictive analysis (or predictive monitoring),
often conflicting goals, crucially depend upon 
the space $S_\tr$ that the analysis reasons about.
In the most general case, $S_\tr$ can be the set all
executions that preserve the control and data flow of $\tr$,
namely \emph{correct reorderings}~\cite{cp2012} of $\tr$.
Analyses that exhaustively reason about the entire space of
correct reorderings have the highest prediction power in theory~\cite{MaximalCausalModel2013},
but quickly become intractable even for very simple classes of bugs~\cite{Mathur2020,wcp2017}.
On the other extreme is the class of trivial analyses which consider 
$S_\tr = \set{\tr}$ but offer no predictive power.
Analyses based on Mazurkiewicz's \emph{trace equivalence} theory~\cite{Mazurkiewicz1987}
opt for middleground and balance predictive 
power with moderate computational complexity of the predictive monitoring question.

In the framework of trace theory, 
one fixes a concurrent alphabet $(\alphabet, \dep)$ consisting of 
a finite set of labels $\alphabet$,
and a symmetric, reflexive dependence relation $\dep \subseteq \alphabet\times\alphabet$.
Now, each string $w \in \alphabet^*$ can be generalized to its
equivalence class $\eqcl{\dep}{w}$, comprising all those strings $w'$
which can be obtained from $w$ by repeatedly swapping neighbouring 
events when their labels are not dependent.
The corresponding predictive monitoring question under trace equivalence 
then translates to the disjointness check $\eqcl{\dep}{\tr} \cap L_{\text{bug}} \neq \emptyset$.
Consider, for example, the execution $\tr_1$ in \figref{race-trace}
consisting of $6$ events $\set{e_i}_{i \leq 6}$ performed by threads $t_1$ and $t_2$.
Consider the \emph{sound }dependence relation $\dep_\rwl$ that
marks pairs of events of the same thread, and pairs of events that write to 
the same memory location or the same lock as dependent.
Here, we say $\dep$ is sound if one can only infer correct reorderings from $\dep$,
i.e., for every well-formed execution $\tr$, $\eqcl{\dep}{\tr} \subseteq \creorder{\tr}$.
It is easy to conclude that $\tr_1$ is equivalent
to the reordering $\rho_1 = e_1e_4e_2e_3e_5e_6$,
i.e., $\rho_1 \in \eqcl{\dep_\rwl}{\tr}$ and thus $\eqcl{\dep}{\tr}$ is not disjoint from 
the set of executions where two $\wt(x)$ events are consecutive.
For a large class of languages 
$L_{\text{bug}}$~\cite{Ochmanski85}, 
this question can, in fact, be answered in a one pass streaming
constant space algorithm,
the holy grail of runtime monitoring,
and has been instrumental in the success of 
industrial strength concurrency bug 
detectors~\cite{threadsanitizer,muehlenfeld2007fault,jannesari2009helgrind+}.
% The algorithmic efficiency in reasoning with Mazurkiewicz traces
% stems from the context insensitivity nature of the orderings
% that a dependence induces --- given two events $e_1$ $e_2$
% in $\tr$, whether or not they can be reordered in an equivalent execution
% solely depends upon the labels of events that occur between them,
% independent of which events precede $e_1$ or succeed $e_2$.

Despite the simplicity and algorithmic efficiency of reasoning
with commutativity of \emph{individual events}, 
trace theory falls short in accurately reason about
commutativity of \emph{atomic} blocks of events in 
executions of concurrent programs. 
Consider, for example, the execution $\tr_2$ from \figref{no-race-trace}.
Here, under the dependence $\dep_\rwl$ described above,
the events $e_1 = \ev{t_1}{\wt(x)}$ and $e_6 = \ev{t_2}{\wt(x)}$
are ordered through the chain of dependence $(e_1, e_3), (e_3, e_4), (e_4, e_6)$.
However, the reordering $\rho'_2 = e_4e_5e_1e_6e_2e_3$ is
a correct reordering of $\tr_2$ and also witnesses the two write events consecutively.
In other words, the equivalence induced by a dependence relation 
can be conservative since commutativity on individual events maybe insufficient
to determine when two blocks commute.
Observe that simply relaxing the dependence relation $\dep_\rwl$
to a smaller set (say by removing dependence on locks)
may be detrimental to soundness as one may infer that the ill-formed execution
$\rho''_2 = e_1e_2e_4e_3e_5e_6$ is equivalent to $\tr_2$.
Indeed, one can show that $\dep_\rwl$ is the most relaxed sound dependence relation.
At the same time, we remark that the efficiency of the algorithms based on
trace equivalence~\cite{goldilocks2007,Djit1999,farzan2008monitoring,fasttrack2009,mathur2020atomicity}
crucially stems from reasoning about commutativity of individual events
(instead of blocks of events).

\input{figures/trace-limitation-race}

% \ucomment{At the same time, this D is the smallest possible dependence relation as otherwise you can break soundness.}
% \ucomment{Indeed, full-fledged reasoning about commutativity of blocks of code without interleaving them is in general intractable [CITE]}.
% \ucomment{Say that Maz traces give a very syntactic (aka context-free of sorts) treatment of dependencies, that can be inferred just by looking at the labels. We would like to emulate the same feature because it is the main reason to get algorithmic efficiency.}

In this work, we propose taking a different route for 
enhancing the predictive power of trace based reasoning. 
Instead of allowing flexibility for commuting individual
\emph{blocks} of events, 
we observe that we can nevertheless enhance predictive power
by sticking to commutativity of events
but allowing for \emph{greater flexibility in selecting
events that participate in these reorderings.}
Consider, for example, the execution $\rho_2$ in \figref{race-correct-reordering},
which is a correct reordering of $\tr_2$ and also witnesses that
the two write events are consecutive.
Intuitively, one can obtain $\rho_2$ from $\tr_2$
by first dropping the earlier critical section in $t_1$,
thereby unblocking the critical section in $t_2$
so that it can commute to the beginning of the execution using event based commutativity.

In this work, we argue that the above style of reasoning can be formalized
as a simple extension of the classic trace theory and no sophisticated
algebraic formulation may be required.
The dependence relation $\dep$ plays a dual 
role --- ({\bf downward-closure}) for each event $e$ in some reordering, 
all events dependent before $e$ must be present in 
the reordering, and 
 ({\bf order-preservation}) amongst the set of events present in the reordering,
the relative order of all dependent events must be preserved.
Towards this, we propose to make this distinction explicit.
Reflecting on the example above, a key tool we employ here is to
\emph{stratify} the
dependence in $\dep_\rwl$ based on their strength.
On one hand, we have \emph{strong} dependencies,
such as program order, for which both the roles 
({\bf downward-closure}) and ({\bf order-preservation})
must be respected and cannot be relaxed.
On the other hand, we have dependence between lock events,
for which ({\bf order-preservation}) must be kept intact, but
nevertheless the first role ({\bf downward-closure}) can be relaxed.
We formalize this notion, in \secref{def-prefix},
using two sets of dependence relations, a strong dependence 
$\dstrong$ and a weak one $\dweak$, and the resulting 
notion of a \emph{strong trace prefix} of an execution, whose set of
events is downward closed with respect to $\dstrong$ and further,
the relative order on the residual events in it respects the order
induced by $\dstrong \cup \dweak$.

Our generalization of traces to strong trace prefixes
has important advantages.
First, and the most obvious one, is the enhanced  predictive power 
when monitoring
against a language $L_{\text{bug}}$, as we illustrated above.
The second consequence of the 
explicit stratification of the dependence relation, 
is that we can predict against new, previously impossible, 
languages such as those
for deadlock prediction~\cite{tuncc2023sound}.
% \ucomment{We elaborate more on this in ...}
Third, the simplicity of our strong trace prefixes framework and its proximity to the original trace-theoretic framework implies that 
the predictive monitoring question in this new setting is solvable in essentially the same time and space complexity as in the trace-theory setting,
despite the enhanced predictive power it unveils.
We present a unified scheme, in \secref{complexity}, to translate
any predictive algorithm that works under trace equivalence against some language $L_{\text{bug}}$
to one that works under strong trace prefixes (for the same language $L_{\text{bug}}$)
with additional non-determinism (but similar time and space usage)
or alternatively with a polynomial multiplicative blowup in time.
Thus, when the predictive question can be answered in constant space for 
Mazurkiewicz traces (as with data races~\cite{ang2023predictive}), 
it continues to be solvable in constant space for strong trace prefixes.

In \secref{rf} we further shorten the gap between
commutativity style reasoning (aka strong trace prefixes) and
the full semantic space (aka correct reorderings).
% by offloading the burden of soundness from the dependence
% onto the algorithms for predictive monitoring.
In particular, we show that we can further relax the
dependence on conflicting memory locations
($(\rd, \wt)$, $(\wt, \rd)$),
that otherwise ensure soundness, and regain soundness back
by baking in extra \emph{reads-from} constraints in the prefixes.
We define \emph{strong reads-from prefixes} to formalize
the resulting space of reorderings, and show that predictive monitoring
under them can also be done with same time and space complexity as with strong trace prefixes.
Next, in \secref{syncp} we draw an interesting connection between strong reads-from prefixes 
and the class of \emph{synchronization-preserving}
data races~\cite{SyncP2021} and deadlocks~\cite{tuncc2023sound}
which are the fastest and most predictive practical algorithms 
for detecting these concurrency bugs.
We show that while synchronization preserving reorderings are a
larger class of reorderings, strong reads-from prefixes are nevertheless
sufficient to capture the corresponding this class of data races and deadlocks.
As a consequence, we obtain constant space algorithms
for predicting class of bugs, improving the previously known linear space algorithms.

We put the new formalism to test by 
implementing the algorithms that follow from our results,
for various specifications such as data races, 
deadlocks and pattern languages~\cite{ang2023predictive}.
We evaluated them on benchmark program traces derived from real world 
software applications and demonstrate the effectiveness of our formalism
through its enhanced prediction power.
% \ucomment{ADD MORE DETAILS.}

%% file: figures/trace-limitation-race.tex
%!TEX root=../main.tex

\begin{figure}[t]
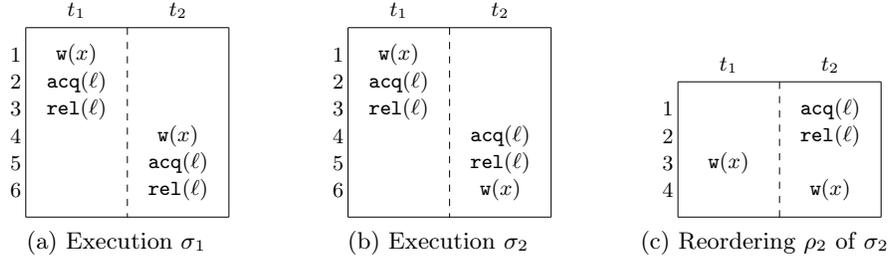

\centering
\subfloat[Execution $\tr_1$ \figlabel{race-trace}]{%
	\execution{2}{
		\figev{1}{$\wt(x)$}
		\figev{1}{$\acq(\lk)$}
		\figev{1}{$\rel(\lk)$}
		\figev{2}{$\wt(x)$}
		\figev{2}{$\acq(\lk)$}
		\figev{2}{$\rel(\lk)$}
	}
}
\hfill
\subfloat[Execution $\tr_2$ \figlabel{no-race-trace}]{%
	\execution{2}{
		\figev{1}{$\wt(x)$}
		\figev{1}{$\acq(\lk)$}
		\figev{1}{$\rel(\lk)$}
		\figev{2}{$\acq(\lk)$}
		\figev{2}{$\rel(\lk)$}
		\figev{2}{$\wt(x)$}
	}
}
\hfill
\subfloat[Reordering $\rho_2$ of $\tr_2$ \figlabel{race-correct-reordering}]{
	\execution{2}{
		\figev{2}{$\acq(\lk)$}
		\figev{2}{$\rel(\lk)$}
		\figev{1}{$\wt(x)$}
		\figev{2}{$\wt(x)$}
	}
}
\caption{Execution $\tr_1$ has a predictable data race, and can be exposed with trace equivalence.
Execution $\tr_2$ has a predictable data race (witnessed by $\rho_2$) which cannot be
exposed under trace equivalence, but can be exposed by strong trace prefixes.}
\figlabel{trace-limitation-race}
\vspace{-0.2in}
\end{figure}

%% file: prelim.tex
%!TEX root=./main.tex

\section{Predictive Monitoring and Trace Theory}
\seclabel{prelim}

Here we discuss some preliminary background on the predictive monitoring problem,
trace theory and some limitations when applying the latter in the context of the former. 

\myparagraph{Events, executions and monitoring}{
We model an execution as a finite sequence $\tr = e_1, e_2, \ldots, e_k$ 
of events where each event $e_i$ is labelled with a letter $a_i = \lbl{e_i} \in \alphabet$
from a fixed alphabet $\alphabet$.
We will use $\events{\tr}$ to denote the set of events of $\tr$
and use the notation $e_1 \trord{\tr} e_2$ to denote that the event
$e_1$ appears before $e_2$ in the sequence $\tr$.
We will often use the custom alphabet $\alphabet_\rwl$ to label events of 
shared memory multithreaded programs.
For this, we fix sets $\threads$, $\locks$ and $\vars$
of thread, lock, and memory location identifiers.
Then, $\alphabet_\rwl = \setpred{\ev{t}{\op(d)}}{t \in \threads, \op(d) \in \set{\rd(x), \wt(x), \acq(\lk), \rel(\lk)}_{x \in \vars, \lk \in \locks}}$
consists of labels denoting
read/write of
memory locations $\vars$ or acquire/release of locks $\locks$,
each being performed by some thread $t \in \threads$.
Executions of multithreaded programs are assumed to be well-formed, i.e., belong to the regular language $L_\wf \subseteq \alphabet_\rwl^*$ 
that contains all strings where each release event $e$ has a unique matching acquire event $e'$ on the same lock and same thread, and no two critical sections on the same lock overlap.
Our focus here is the runtime monitoring problem against a property 
$L \subseteq \alphabet^*$ --- `given an execution $\tr$, does $\tr \in L$?'
% and is also the focus of our work.
}

\myparagraph{Predictive monitoring and correct reorderings}{
% The effectiveness of runtime monitoring for testing concurrent programs
% can often be adversely affected by sensitivity to thread scheduling ---
% even when the program in consideration has a buggy execution, the
% uncontrolled thread scheduling may never expose the right execution that
% witnesses the buggy behavior.
% Thanks to non-deterministic thread scheduling,
Vanilla dynamic analyses that answer the membership question `$\tr \in L$?' 
often miss bugs thanks to non-deterministic thread interleaving.
Nevertheless, even when an execution $\tr$ does not belong to the target language $L$,
it may still be possible to \emph{predict} bugs in alternate executions that can be inferred from $\tr$.
% that is infer the presence of bugs in the underlying programs
% by considering alternate interleavings of $\tr$.
Here, one first defines the set $\creorder{\tr}$ of 
\emph{correct reorderings}~\cite{cp2012,MaximalCausalModel2013}
of $\tr$ comprising executions similar to $\tr$ in the following
precise sense --- every program $P$ that can generate $\tr$, will also generate
all executions in $\creorder{\tr}$.
% Having access to the set $\creorder{\tr}$ allows one to answer the 
% \emph{predictive monitoring} question against a 
% language $L$ --- `given an execution $\tr$, is $\creorder{\tr} \cap L \neq \emptyset$?'.
For an execution $\tr \in L_\wf \subseteq \alphabet^*_\rwl$ 
of a multithreaded program, $\creorder{\tr}$ 
can be defined to be the set of all executions $\rho$ 
of $\tr$ such that (1) $\events{\rho} \subseteq \events{\tr}$,
(2) $\rho$ is well-formed, i.e., $\rho \in L_\wf$,
(3) $\rho$ is downward-closed with respect to the \emph{program-order} of $\tr$, 
i.e., for any two events $e_1 \trord{\tr} e_2$, if $e_2 \in \events{\rho}$, then $e_1 \in \events{\rho}$, and
(4) for any read event $e_{\rd}$ labelled $\ev{t}{\rd(x)}$, 
the write event $e_{\wt}$
that $e_{\rd}$ reads-from ($e_{\rd} \in \rf{\tr}{e_{\wt}}$)
must also be in $\rho$.
Here, we say that $e_{\rd} \in \rf{\tr}{e_{\wt}}$
if $e_{\rd}$ and $e_{\wt}$ access the same memory location $x$
and there is no other write event $e'_{\wt}$ such that $e_{\wt} \trord{\tr} e'_{\wt} \trord{\tr} e_\wt$.
The \emph{predictive monitoring} question against a 
language $L$ can now be formalized as `given an execution $\tr$, is $\creorder{\tr} \cap L \neq \emptyset$?'.
% Clearly, every $\rho \in \creorder{\tr}$ passes the same control flow
% as $\tr$ and thus any program that generates $\tr$ can also generate $\rho$.
Observe that any witness $\rho \in \creorder{\tr} \cap L$ is a true positive
since every execution in $\creorder{\tr}$ passes the same control flow
as $\tr$ and thus can be generated by any program that generates $\tr$.
In general, this predictive monitoring question
does not admit a tractable solution, even for the simplest class of (regular) languages,
such as the class of executions that contain a data race~\cite{Mathur2020}, 
and has been shown to admit super-linear space hardness even for $2$ threads~\cite{grain2024}.
Practical and sound algorithms for solving the predictive monitoring 
problem~\cite{wcp2017,SHB2018,SyncP2021,Pavlogiannis2019,cp2012} 
often weaken predictive power in favour of soundness by considering a smaller space $S_\tr$ of reorderings.
A set $S_\tr \subseteq \alphabet^*_\rwl$ is said to be
\emph{\underline{sound} for a given execution $\tr \in L_\wf$} if 
$S_\tr \subseteq \creorder{\tr}$, an algorithm that 
restricts its search of reorderings to $S_\tr$
will not report false positives.
}

\myparagraph{Mazurkiewicz traces}{
Trace theory, proposed by Antoni Mazurkiewicz~\cite{Mazurkiewicz1987},
offers a tractable solution to the otherwise intractable predictive monitoring problem,
by characterizing a simpler subclass of reorderings.
Here, one identifies a reflexive and symmetric \emph{dependence} 
relation $\dep \subseteq \alphabet \times \alphabet$, and deems an execution
$\rho$ equivalent to $\tr$ if one can obtain $\rho$ from $\tr$ by
repeatedly swapping neighbouring events when they are not dependent.
Together, $(\alphabet, \dep)$ constitute a \emph{concurrent alphabet}.
Formally, the \emph{trace equivalence} $\mazeq{\dep}$ of the concurrent alphabet $(\alphabet, \dep)$ 
is the smallest
equivalence on $\alphabet^*$ such that for every $w_1, w_2 \in \alphabet^*$
and for every $(a, b) \in \alphabet \times \alphabet \setminus \dep$,
we have 
% \[w_1 \cdot a \cdot b \cdot w_2 \mazeq{\dep} w_1 \cdot b \cdot a \cdot w_2.\]
$w_1 \cdot a \cdot b \cdot w_2 \mazeq{\dep} w_1 \cdot b \cdot a \cdot w_2.$
We use $\eqcl{\dep}{w} = \setpred{w'}{w \mazeq{\dep} w'}$ to denote the equivalence
class of $w \in \alphabet^*$.
}

\myparagraph{Modelling shared-memory concurrency using traces}{
% Recall that the set of correct reorderings satisfy certain well-formedness
% and control-flow preserving properties.
Let us see how traces can (conservatively) model 
a class of correct reorderings, with an appropriate choice of dependence over
$\alphabet_\rwl$.
The dependence $\dep_{\locks} = 
\setpred{(\ev{t_1}{\op_1(\lk)}, \ev{t_2}{\op_2(\lk)})}{\lk \in \locks}$ 
can be used enforce mutual exclusion of critical sections --- for every 
$\rho \in \eqcl{\dep_{\locks}}{\tr}$,
the order of locking events is the same as in $\tr$, and thus
if $\tr$ is well-formed, then so is $\rho$.
Likewise, the dependence $\dep_{\threads} = \setpred{(\ev{t}{\op_1(d_1)}, \ev{t}{\op_2(d_2)})}{t \in \threads}$
is such that every $\rho \in \eqcl{\dep_{\threads}}{\tr}$ preserves the
program order of $\tr$.
Indeed, the dependence $\dep_{\mathsf{HB}} = \dep_{\threads} \cup \dep_{\locks}$
is the classic \emph{happens-before} dependence
employed in modern data race detectors such as \tsan~\cite{threadsanitizer}.
Finally, the dependence
$\dep_{\rwl} = \dep_{\threads} \cup \dep_{\locks} \cup \dep_{\mathsf{conf}}$,
where
$\dep_{\mathsf{conf}} = \setpred{(\ev{t_1}{\op_1(x)}, \ev{t_2}{\op_2(x)})}
{
x \in \vars, \land (\op_1, \op_2) \in \set{(\wt, \rd), (\rd, \wt), (\wt, \wt)}
}$ ensures that for a well-formed execution $\tr$,
we have $\eqcl{\dep_{\rwl}}{\tr} \subseteq \creorder{\tr}$.
The inclusion of $\dep_{\threads}$ ensures that program order is preserved,
$\dep_{\locks}$ ensures well-formedness, while the remaining dependencies
preserve the order of all conflicting pairs of events, and thus the reads-from relation.
Indeed, $\dep_\rwl$ is the smallest dependence that ensures soundness.
Here, we say that $\dep \subseteq \alphabet_\rwl \times \alphabet_\rwl$ is 
sound if for every $\tr \in \alphabet^*$, $\eqcl{\dep}{\tr} \subseteq \creorder{\tr}$,
}

\myparagraph{Predictive monitoring with Traces}{
The predictive monitoring
question under trace equivalence 
induced by a generic concurrent alphabet 
$(\alphabet, \dep)$ becomes --- `given an execution $\tr$, is $\eqcl{\dep}{\tr} \cap L \neq \emptyset$?'.
In general, even when $L$ is regular, this problem cannot be solved faster than $O(|\tr|^\wdth)$.
Here, $\wdth$ is the degree of concurrency in $\dep$, or
the size of the largest set without containing pairwise dependent events~\cite{ang2023predictive}.
For the subclass of \emph{star-connected} regular languages~\cite{Ochmanski85},
this problem can be solved using a constant-space linear-time algorithm.
Star-connected languages include the class of languages that 
can encode data races~\cite{Djit1999,goldilocks2007},
and the class of \emph{pattern languages}~\cite{ang2023predictive} that capture
other temporal bugs.
 % in concurrent software.
}

% \vspace{-0.1in}
\begin{example}
\exlabel{expose-races-traces}
Let $L_{\textsf{race}} = \alphabet_{\rwl}^* \ev{t_1}{\wt(x)} \ev{t_2}{\wt(x)} \alphabet_{\rwl}^*$
be the set of executions that witness a race between
two write accesses on memory location $x$
between threads $t_1$ and $t_2$.
Consider the execution $\tr_1$ illustrated in~\figref{race-trace}
and recall from \secref{intro} that 
$\eqcl{\dep_\rwl}{\tr_1} \cap L_{\textsf{race}} \neq \emptyset$.
% Let us use $e_i$ to denote the $i^{\text{th}}$ event from the top.
% Observe that $\tr \not\in L_{\textsf{race}}$.
% However,
% $\rho_1 = \ev{t_1}{\wt(x)} \ev{t_2}{\wt(x)} \ev{t_1}{\acq(\lk)} \ev{t_1}{\rel(\lk)} \ev{t_2}{\acq(\lk)} \ev{t_2}{\rel(\lk)}$ belongs to
% $L_{\textsf{race}}$, and further
% $\rho_1 \in \eqcl{\dep_\rwl}{\tr_1} \subseteq \creorder{\tr_1}$.
% In other words, $\eqcl{\dep_\rwl}{\tr_1} \cap L_{\textsf{race}} \neq \emptyset$
% and thus, trace equivalence can be used to expose the data race in $\tr_1$.
Further, recall that for the trace $\tr_2$ from \figref{no-race-trace},
we have $\eqcl{\dep_\rwl}{\tr_2} \cap L_{\textsf{race}} = \emptyset$,
even though $\creorder{\tr_2} \cap L_{\textsf{race}} \neq \emptyset$.
In other words, data race prediction based on
trace equivalence may have strictly less predictive power
than prediction based on correct reorderings.
\end{example}
% \begin{example}
% \exlabel{miss-races-traces}
% Let us now consider the execution $\tr_2$ from~\figref{no-race-trace}.
% Here again, $\tr_2 \not\in L_{\textsf{race}}$.
% However observe that
% $\eqcl{\dep_{\rwl}}{\tr} \cap L = \emptyset$.
% This is because the consecutive pairs of labels
% $(\lbl{e_i}, \lbl{e_{i+1}}) \in \dep_\threads \cup \dep_\locks$ for each $1\leq i\leq 5$.
% While trace equivalence cannot be used to infer the presence of a data race,
% the execution $\rho_2 \in L_{\textsf{race}}$ from \figref{race-correct-reordering}
% nevertheless witnesses the predictable race $(\ev{t_1}{\wt(x)}, \ev{t_2}{\wt(x)})$
% because  $\rho_2 \in \creorder{\tr_2}$.
% In other words, data race prediction based on
% trace equivalence may have strictly less predictive power
% than prediction based on the more general set of correct reorderings.
% \end{example}
\input{figures/trace-limitation-deadlock}

\vspace{-0.45cm}
\begin{example}
\exlabel{cannot-model-deadlock-traces}
While trace equivalence can expose some data races (as with $\tr_1$ from~\figref{race-trace}),
it can fundamentally not model deadlock prediction.
Consider the execution $\tr_3$ in \figref{no-deadlock-trace}.
It consists of two nested critical sections in inverted order of acquisition.
Any program that generates $\tr_3$ is prone to a deadlock, as witnessed by the correct
reordering $\rho_3$ in \figref{deadlock-correct-reordering}
that acquires $\lk_1$ in $t_1$ and then immediately
switches context to $t_2$ in which lock $\lk_2$ is acquired.
Clearly, the underlying program is deadlocked at this point.
Since $(\lbl{e_3}, \lbl{e_5})$ and $(\lbl{e_4}, \lbl{e_6}) \in \dep_\locks$,
trace equivalence cannot predict this deadlock.
Indeed,
nested critical sections,
acquired in a cyclic order,
can never be reordered to actually expose the deadlock
without violating the dependence between earlier release
events and later acquire events induced by $\dep_\locks$.
\end{example}

%% file: figures/trace-limitation-deadlock.tex
%!TEX root=../main.tex

\begin{figure}[t]
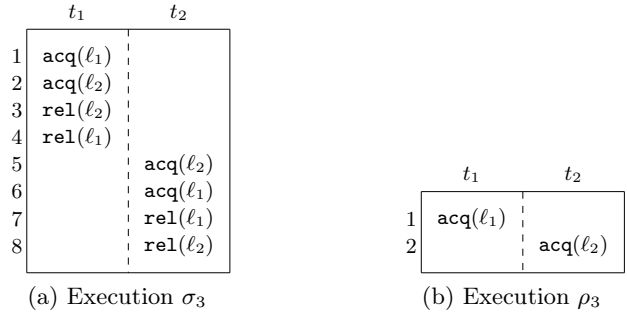

\centering
\subfloat[Execution $\tr_3$ \figlabel{no-deadlock-trace}]{%
	\execution{2}{
		\figev{1}{$\acq(\lk_1)$}
		\figev{1}{$\acq(\lk_2)$}
		\figev{1}{$\rel(\lk_2)$}
		\figev{1}{$\rel(\lk_1)$}
		\figev{2}{$\acq(\lk_2)$}
		\figev{2}{$\acq(\lk_1)$}
		\figev{2}{$\rel(\lk_1)$}
		\figev{2}{$\rel(\lk_2)$}
	}
}
\hfil
\subfloat[Execution $\rho_3$ \figlabel{deadlock-correct-reordering}]{%
	\execution{2}{
		\figev{1}{$\acq(\lk_1)$}
		\figev{2}{$\acq(\lk_2)$}
	}
}
\caption{Execution $\tr_3$ has a predictable deadlock, as witnessed
by the correct reordering $\rho_3$, but cannot be exposed by
without violating $\dep_\locks$.}
\figlabel{trace-limitation-deadlock}
\vspace{-0.3in}
\end{figure}

%% file: prefix-def.tex
%!TEX root=./main.tex

\section{Strong Trace Prefixes}
\seclabel{def-prefix}

Observe that for both the executions
$\tr_2$ (\exref{expose-races-traces}) and 
$\tr_3$ (\exref{cannot-model-deadlock-traces}), 
the correct reordering that exposes
the bug in question can be obtained by relaxing the order of
two events that were otherwise ordered by the dependence relation,
in particular $\dep_\locks$.
Since the dependence $\dep_\locks$
enforces mutual exclusion, it cannot be ignored altogether
without compromising soundness.
For example, setting $\dep_\locks = \emptyset$,
would deem $\rho_2' = \ev{t_1}{\wt(x)} \ev{t_1}{\acq(\lk)}  \ev{t_2}{\acq(\lk)} \ev{t_1}{\rel(\lk)}  \ev{t_2}{\rel(\lk)} \ev{t_2}{\wt(x)}$ to be equivalent to $\tr_2$,
even though $\rho'_2 \not\in \creorder{\tr_2}$.
Nevertheless, both these examples
illustrate a key insight behind how we generalize 
the trace-theoretic framework --- the dependence due to locks is \emph{weak}.
That is, let $e_1 = \ev{t_1}{\rel(\lk)} \trord{\tr} e_2 = \ev{t_2}{\acq(\lk)}$
be events of an execution $\tr$.
If they both appear in a reordering $\rho$ of $\tr$, 
then, under commutativity-style reasoning, we demand that their 
relative order must be $e_1 \trord{\rho} e_2$. 
However, reorderings that drop the entire critical section of $e_1$
may nevertheless be allowed and may not compromise well-formedness.
This is in contrast with \emph{strong} dependence such as those induced  
due to $\dep_\threads$ or \emph{reads-from} --- any reordering must be downward closed
with respect to them.

Building on these insights, we formalize strong trace prefixes by distinguishing
dependencies that are absolutely necessary, i.e., \emph{strong} dependence,
from \emph{weak} dependence, which do not affect \emph{causality}, but only offer
convenience for modelling
 constructs like mutual exclusion in a swap-based equivalence 
like trace equivalence.
We present the formal definition of strong trace prefixes next.

\begin{definition}[Dual Concurrent Alphabet]
\deflabel{dual-conc-alphabet}
A \emph{\underline{dual concurrent alphabet}} is a tuple
$(\alphabet, \dstrong, \dweak)$, where
$\alphabet$ is a finite alphabet,
$\dstrong \subseteq \alphabet \times \alphabet$ is a reflexive and symmetric
\emph{strong} dependence relation, and
$\dweak \subseteq \alphabet \times \alphabet$ is an irreflexive symmetric
\emph{weak} dependence relation.
\end{definition}

\begin{definition}[Strong Trace Prefix]
\deflabel{strong-trace-prefix}
The \emph{\underline{strong trace prefix order}}  induced by
the dual alphabet $(\alphabet, \dstrong, \dweak)$,
denoted $\strongqo{\dstrong}{\dweak}$, is the smallest 
reflexive and transitive binary relation on $\alphabet^*$ that satisfies:
\vspace{-0.2cm}
\begin{enumerate}
	\item $\mazeq{\dstrong \cup \dweak} \; \subseteq \; \strongqo{\dstrong}{\dweak}$, and
	\item for every $u, v \in \alphabet^*$ and for every $a \in \alphabet$,
	if for every $b \in v$, $(a, b) \not\in \dstrong$, then we have
	$u \cdot v \strongqo{\dstrong}{\dweak} u \cdot a \cdot v$
\end{enumerate}
\vspace{-0.2cm}
We say that $w' \in \alphabet^*$ is a \emph{\underline{strong trace prefix}} of $w$
if $w' \strongqo{\dstrong}{\dweak} w$.
We use $\strongdcl{w}{\dstrong}{\dweak} = \setpred{w'\in\alphabet^*}{w' \strongqo{\dstrong}{\dweak} w}$ to denote the \emph{\underline{strong downward closure}} of $w$.
\end{definition}

\noindent
Let us also recall the classical notion of ideal based prefixes
using the above.
For a reflexive symmetric dependence relation 
$\dep \subseteq \alphabet \times \alphabet$, 
we use the notation $\idealqo{\dep}$ to denote the \emph{ideal prefix} relation
$\strongqo{\dep}{\emptyset}$, and call $w_1$ an \emph{\underline{ideal prefix}} of $w_2$
if $w_1 \idealqo{\dep} w_2$.
We use $\idealdcl{w}{\dep} = \setpred{w' \in \alphabet^*}{w' \idealqo{\dep} w}$ to 
denote the \emph{\underline{ideal downward closure}} of $w$. 

A few observations about \defref{strong-trace-prefix} are in order.
First, the relations $\strongqo{\dstrong}{\dweak}$ and $\idealqo{\dep}$ defined here
are not equivalence relations (unlike $\mazeq{\dep}$)
but only quasi orders and relate executions of different lengths (namely strong (or ideal) prefixes).
Second, in the case $\dweak \subseteq\dstrong$, the strong trace prefix
order gives the ideal prefix order $\idealqo{\dstrong \cup \dweak}$.
% i.e., $w_1 \strongqo{\dstrong}{\emptyset} w_2$ iff there is a
% $w'_2$ such that $w'_2 \mazeq{\dstrong} w_2$ and $w_1$ is a prefix of $w'_2$ (i.e., there is a $v$ such that $w'_2 = w_1 \cdot v$).
Third, in general, strong prefixes are more permissive than
ideal prefixes, i.e.,  $\strongqo{\dstrong \cup \dweak}{\emptyset} \subseteq \strongqo{\dstrong}{\dweak}$,
and is key to enhancing the predictive power  of commutativity-style reasoning.

\begin{example}
% Let us illustrate strong trace prefixes using a concrete example.
Consider the alphabet $\alphabet = \set{a, b, c}$.
Let
$\dstrong = \set{(a, a), (b, b), (c, c), (b, c), (c, b)}$
be the strong dependence 
and let $\dweak = \set{(a, b), (b, a)}$ be the weak dependence.
Let $\dep = \dstrong \cup \dweak$ be a traditional Mazurkiewicz-style dependence.
Now, consider the string $w = abacba$.
First, observe the simple equivalence $w \mazeq{\dep} w' = abcaba$. 
Indeed, no other strings in $\alphabet^*$ are $\mazeq{\dep}$-equivalent to $w$.
The ideal prefixes of $w$
$\idealdcl{w}{\dep} = \set{\epsilon, a, ab, aba, abac, abacb, abacba, abc, abca, abcab, abcaba}$
is precisely the set of (string) prefixes of the two strings $w$ and $w'$.
The set of strong trace prefixes induced by $(\alphabet, \dstrong, \dweak)$
is larger though.
First, consider the string $w_1 = abcb$ and observe that
$w_1 \strongqo{\dstrong}{\dweak} w$.
This follows because 
% (1) $w = \underbrace{abcab}_{u_1}a\underbrace{\epsilon}_{v_1}$,
(1) $w = abcab\cdot a\cdot\epsilon$,
and thus $abcab \strongqo{\dstrong}{\dweak} w$, 
% (2) $abcab = \underbrace{abc}_{u_2}a\underbrace{b}_{v_2}$
(2) $abcab = abc\cdot a\cdot b$
and $(a, b) \not\in \dstrong$ and thus $abcb \strongqo{\dstrong}{\dweak} abcab$,
and finally (3) due to transitivity, we have $w_1 \strongqo{\dstrong}{\dweak} w$.
Consider now the string $w_2 = bcb$, and observe that
% $abcb = \underbrace{\epsilon}_{u_3}a\underbrace{bcb}_{v_3}$, 
$abcb = \epsilon\cdot a\cdot bcb$,
giving us $w_2 \strongqo{\dstrong}{\dweak} w_1$
since $\set{(a,b), (a,c)} \cap \dstrong = \emptyset$.
Thus, $w_2 \strongqo{\dstrong}{\dweak} w$.
On the other hand, observe that $w_1 \nidealqo{\dep} w$ and $w_2 \nidealqo{\dep} w$.
\end{example}

\subsection{Modelling correct reorderings with strong trace prefixes}

% Let us now demonstrate how to instantiate
% the abstract framework of strong trace prefixes,
% in the context of shared memory multi-threaded programs. % for predictive monitoring.
Recall that $\dep_\rwl$ ordered events of the same thread,
same locks and conflicting events of a given memory location
allowing us to soundly represent a class of correct reorderings of an execution $\tr$ 
as the equivalence class $\eqcl{\dep_\rwl}{\tr}$.
Here, we identify a finer gradation of $\dep_\rwl$,
to allow for a larger subset of correct reorderings.
Specifically, we define the strong and weak dependence on $\alphabet_\rwl$ as:
\vspace{-0.3cm}
\begin{align}
\begin{array}{c}
% \dweak_{\wt(\vars)} 
\dweak_{\wt} 
= \setpred{(\ev{t_1}{\wt(x)}, \ev{t_2}{\wt(x)})}
{x \in \vars, t_1\neq t_2 \in \threads},
\quad
% &\quad
% &
\dweak_\locks = \setpred{(a^{t_1}, b^{t_2}) \in \dep_\locks}{t_1\neq t_2}\\
\dweak_{\rwl} = \dweak_\locks \cup \dweak_{\wt},
% &\quad
% &
\quad\quad
\dstrong_\rwl = \dep_\rwl \setminus \dweak_\rwl
\end{array}
\vspace{-0.2cm}
\equlabel{rwl}
\end{align} 
% \[\dweak_{\wt(\vars)} = \setpred{(\ev{t_1}{\wt(x)}, \ev{t_2}{\wt(x)})}
% {x \in \vars\text{ and }t_1 \neq t_2} \text{ and}\]
% \[\dweak_\locks = \setpred{(a, b) \in \dep_\locks}{a \neq b},\]
% \[\dweak_{\rwl} = \dweak_\locks \cup \dweak_{\wt(\vars)}.\]
% Similarly, let $\dstrong_\rwl = \dep_\rwl \setminus \dweak_\rwl$.
In other words, the dual concurrent alphabet 
$(\alphabet_\rwl, \dstrong_\rwl, \dweak_\rwl)$
relaxes the `hard ordering' between writes to the same memory location (i.e., `conflicting writes')
as well as that between critical sections of the same lock (i.e., `conflicting lock events').
We next explain the intuition behind the above relaxations.

\myparagraph{Weakening dependence on writes}{
Let us begin by arguing about $\dweak_{\wt}$.
When an execution contains two \emph{consecutive} write events 
$e_1, e_2$ with
$\lbl{e_1} = \ev{t_1}{\wt(x)}$ and $\lbl{e_2} =\ev{t_2}{\wt(x)}$ 
on the same memory location $x \in \vars$, 
then, clearly, there is no event reading from the 
first write event $e_1$ since it is immediately overwritten by $e_2$.
In this case, while flipping the order of $e_1$ and $e_2$ may 
violate the read-from relation of a read event reading from $e_2$, 
observe that $e_1$ can be completely \emph{dropped}  
(in absence of later $\dstrong_\rwl$-dependent events after $e_1$)
without dropping $e_2$ and without affecting any control flow.
In other words, the presence of $e_2$ does not mandate the presence of $e_1$,
but when both are present, the conservative choice of placing $e_1$ before $e_2$
ensures that the reads-from relation is preserved.
% Owing to this insight, $\rho'_5$ in \figref{trace-limitation-ww} can be obtained from $\tr_5$ by excluding $\ev{t_1}{\wt(x)}$ from the prefix because $(\ev{t_1}{\wt(x)}, \ev{t_2}{\wt(x)})\in \dweak_\rwl$.
}
% \zacomment{Example here}

\myparagraph{Weakening dependence on lock events}{
% Let us now discuss why we can weaken the dependence on locks.
Recall that the primary role of the dependence $\dep_\locks$ 
was to ensure mutual exclusion, i.e., two critical sections of the same lock do not overlap
in any execution obtained by repeatedly swapping neighboring independent events.
We identify that this is not a strong dependence, 
in that one can possibly \emph{drop} an earlier critical section entirely,
while retaining a later critical section on the same lock in a 
candidate correct reordering. 
The correct reordering $\rho_2$ of $\tr_2$ in
\figref{race-correct-reordering} %(from \exref{miss-races-traces})
can be obtained by leveraging this insight.
Indeed, $\rho_2 \in \strongdcl{\tr_2}{\dstrong_\rwl}{\dweak_\rwl}$
because $(\ev{t_1}{o_1(\lk)}, \ev{t_2}{o_2(\lk)}) \in \dweak_\rwl$
for $o_1, o_2 \in \set{\acq, \rel}$.
Moreover, in the deadlock example (\figref{trace-limitation-deadlock}), $\rho_3 \in \strongdcl{\tr_3}{\dstrong_\rwl}{\dweak_\rwl}$, 
since the critical section of $l_2$ in thread $t_1$ can be completely dropped without affecting the presence of $\acq(l_2)^{t_2}$.
}

\myparagraph{Well-formedness}{
% Recall, however, that completely ignoring $\dep_\locks$ may 
% result in strong trace prefixes where two complete critical sections on the same lock overlap.
The weak dependence $\dweak_\locks$ ensures that no two \emph{complete} 
critical sections on the same lock overlap in a strong trace prefix 
(provided they did not overlap in the original execution).
However, simply marking lock dependencies as \emph{weak} 
still does not forbid strong trace prefixes where an 
earlier \emph{incomplete} critical section overlaps 
with a later \emph{complete} critical section.
Consider for example, the (ill-formed) execution 
$\rho'_{2} = \ev{t_1}{\wt(x)}\ev{t_1}{\acq(\lk)}\ev{t_2}{\acq(\lk)}\ev{t_2}{\rel(\lk)}\ev{t_2}{\wt(x)}$.
Observe that $\rho'_2$ is a strong trace prefix of
$\tr_2$ under $\dstrong_\rwl$ and $\dweak_\rwl$.
% since the set of its events is downward closed with respect
% to $\dstrong_\rwl$ and both strong and weak dependencies are preserved on whatever events remain.
As we will show in \secref{well-formedness}, we can remedy
this mild peculiarity in the predictive monitoring algorithm.
% based on the observation that
% the set $L_\wf$ is (a) regular, and 
% (b) closed under strong prefixes, i.e.,
% for every $\tr \in L_\wf$, we have $\strongdcl{\tr}{\dstrong_\rwl}{\dweak_\rwl} \subseteq L_\wf$,
% and algorithms for predictive monitoring can be easily augmented
% to reason about the set 
% $\strongdcl{\tr}{\dstrong_\rwl}{\dweak_\rwl}\cap L_\wf$.
% This mild shortcoming about incomplete critical sections,
% however, can be remedied because of the observation that
% the set $L_\wf$ is regular, and closed under strong prefixes 
% (i.e., for every $w \in L_\wf$, $\strongdcl{L_\wf}{\dstrong_\rwl}{\dweak_\rwl} \subseteq L_\wf$).
% by equipping our predictive monitoring algorithms
% so that the strong trace prefixes these algorithms reason about
% are also necessarily well-formed.
% In view of this, we will consider the set 
% $\strongdcl{\tr}{\dstrong_\rwl}{\dweak_\rwl}\cap L_\wf$ in the context of 
% predictive monitoring.
}

\myparagraph{Soundness and Precision Power}{
% Let us now comment about the soundness of our choice of dual concurrent alphabet.
% Recall that a set $S_\tr \subseteq \alphabet^*_\rwl$ is
% \emph{sound for a given execution $\tr \in L_\wf$} if 
% $S_\tr \subseteq \creorder{\tr}$.
Strong trace prefixes retain soundness (as long as they are well-formed)
while enjoying higher predictive power:
 % than trace equivalence and ideal prefixes:
\begin{restatable}{theorem}{soundnessStrongPrefix}{\it\bf(Soundness and Precision Power)}.
\thmlabel{soundness}
	For each well-formed execution $\tr\in L_\wf$, we have: 
	\[
	\eqcl{\dstrong_\rwl\cup \dweak_\rwl}{\tr} \subseteq \idealdcl{\tr}{\dstrong_\rwl\cup \dweak_\rwl} \subseteq \strongdcl{\tr}{\dstrong_\rwl}{\dweak_\rwl} \cap L_\wf \subseteq \creorder{\tr}.
	\]
Moreover, there is a $\tr\in L_\wf$ for which each of the subset relationships are strict.
\end{restatable}
}

\myparagraph{Maximality}{
Our choice of the dual concurrent alphabet $(\alphabet_\rwl, \dstrong_\rwl, \dweak_\rwl)$
is also the best one amongst the space of sound dual concurrent alphabets
obtained by stratifying $\dep_\rwl$.
Formally,
% We formalize this \emph{maximality} property as follows.
\begin{restatable}{theorem}{maximality}{\it\bf(Maximality)}.
\thmlabel{maximality}
Let $(\alphabet_\rwl, \dstrong, \dweak)$ be a dual concurrent alphabet such that
$\dep_\rwl \subseteq \dstrong \cup \dweak$, and,
for every $\tr \in \alphabet^*_\rwl$, $\strongdcl{\tr}{\dstrong}{\dweak} \cap L_\wf \subseteq \creorder{\tr}$.
Then, $\dstrong_\rwl \subseteq \dstrong$, and thus,
for every $\tr$,
$\strongdcl{\tr}{\dstrong}{\dweak} \subseteq \strongdcl{\tr}{\dstrong_\rwl}{\dweak_\rwl}$
\end{restatable}
% \begin{theorem}[Maximality]
% \thmlabel{maximality}
% Let $(\alphabet_\rwl, \dstrong, \dweak)$ be a dual concurrent alphabet such that
% $\dep_\rwl \subseteq \dstrong \cup \dweak$, and,
% for every $\tr \in \alphabet^*_\rwl$, $\strongdcl{\tr}{\dstrong}{\dweak} \cap L_\wf \subseteq \creorder{\tr}$.
% Then, $\dstrong_\rwl \subseteq \dstrong$, and thus,
% for every $\tr$,
% $\strongdcl{\tr}{\dstrong}{\dweak} \subseteq \strongdcl{\tr}{\dstrong_\rwl}{\dweak_\rwl}$
% \end{theorem}
}

%% file: complexity.tex
%!TEX root=./main.tex

\section{Complexity of Predictive Monitoring}
\seclabel{complexity}

In this section we investigate the impact of generalizing Mazurkiewicz 
traces to strong trace prefixes, on the predictive monitoring question.
% against a language $L \subseteq \alphabet^*$.
% Our main result is a meta-reduction that relates
% predictive monitoring, against an arbitrary language $L$, under 
% Mazurkiewicz traces equivalence to those predictive monitoring
% strong trace prefixes against the same language $L$.
% We first establish that with additional non-determinism,
% predictive monitoring under strong trace prefixes can be done 
% in essentially the same time and space complexity as Mazurkiewicz traces (\secref{non-det-complexity}).
% We then show, in \secref{det-poly-blowup}, that we can also perform
% predictive monitoring in a deterministic manner
% with only polynomial multiplicative blow-up in time.
We present two \emph{schemes} to translate
arbitrary Turing machines for predictive monitoring under trace equivalence against a language $L$
to one for predictive monitoring under strong trace prefixes against the same language $L$.
The first scheme (\secref{non-det-complexity}), 
uses additional non-determinism (but same time and space usage),
and the second (\secref{det-poly-blowup}) employs
polynomial multiplicative blow-up in time complexity.

\subsection{Non-deterministic predictive monitoring}
\seclabel{non-det-complexity}

\begin{comment}
Let us fix a language $L$ and let $M$ be a deterministic Turing machine
that answers the predictive monitoring question against the language $L$, i.e., $L(M) = \setpred{w}{\eqcl{\dep}{w}\cap L \neq \emptyset}$.
Our first result says that 
there is a nondeterministic Turing machine $M'$ that recognizes the following language $L(M') = \setpred{w}{\strongdcl{w}{\dstrong}{\dweak} \cap L \neq \emptyset}$
and is of the same time and space complexity as $M$.
This result follows from the following observation.
$M'$ nondeterministically pick a subsequence, 
check if the subsequence is a strong trace prefix
and run the Turing machine $M$ on the subsequence.
In order to achieve this, we need only constant extra space but nondeterminism.
\end{comment}

We first show that an algorithm that solves the vanilla predictive monitoring problem
($\eqcl{\dep}{\tr} \cap L \neq \emptyset$) can be transformed into an algorithm
for predictive monitoring against strong trace prefixes
with similar resource (time and space) usage, albeit with use of non-determinism.

\begin{restatable}{theorem}{nondeterminism}
\thmlabel{nondet}
    Let $L \subseteq \alphabet^*$ 
    and let $M$ be a deterministic Turing machine, that uses time $T(n)$ and space $S(n)$, 
    such that $L(M) = \setpred{w}{\eqcl{\dep}{w}\cap L \neq \emptyset}$.
    There is a nondeterministic Turing machine $M'$ that uses time 
    $T(n) + O(n)$ and space $S(n) + O(n)$, 
    such that $L(M') = \setpred{w}{\strongdcl{w}{\dstrong}{\dweak} \cap L \neq \emptyset}$.
    Moreover, if $M$ runs in one-pass,
    then $M'$ uses space $S(n) + c$ (for some constant $c$).
\end{restatable}
Observe that, in the above, $S(n) + c \in O(S(n))$.
Further, when $T(n) \in \Omega(n)$, then 
$T(n) + O(n) \in O(T(n))$.
Thus, the time and space usage of the non-deterministic machine $M'$
in \thmref{nondet} are essentially the same as those of $M$.

The proof of \thmref{nondet} relies on the observation that
any strong prefix $u$ of a string $w$ is equivalent 
(according to trace equivalence using $\mazeq{\dstrong\cup \dweak}$) to a subsequence $w'$ of $w$,
such that $w'$ is downward closed with respect to strong dependencies.
The non-deterministic Turing machine $M'$ 
first non-deterministically guesses a subsequence $w'$ of the input execution $w$,
then, using constant space and an additional forward streaming pass, 
ensures that $w'$ is downward closed with respect to $\dstrong$,
and finally invokes the Turing machine $M$ on the string $w'$.

It follows from~\thmref{nondet} that when the language of $M$ 
is regular, so is the language of $M'$.
% , since
% a non-deterministic (or determinisitc) Turing machine $M'$ that runs in 
% constant space is an acceptor of a regular language.
% In other words, there is a deterministic finite automaton $M''$(obtained by determinising $M'$) that also has the
% same asymptotic time and space utilization.
This means, that when a language $L$ can be predictively monitored
in constant space under trace equivalence (for example data races, deadlocks, or pattern languages~\cite{ang2023predictive}),
then it can also be predictively monitored in constant space,
yet with higher predictive power, under strong trace prefixes!
% In fact, for the purpose of monitoring the class of languages we 
% will discuss in this paper (namely data races, deadlocks and pattern languages~\cite{ang2023predictive}), 
% $M$ will inadvertently be a deterministic finite automaton, i.e., 
% the language $L$ can be predictively monitored in constant space and linear time,
% and thus, we can also predictively monitor $L$ with higher coverage under strong trace prefixes.
% \begin{corollary}
%     Given a language $L\subseteq\alphabet^*$. 
%     If $L$ can be predictively monitored in space $O(1)$ and time $O(n)$ under trace equivalence,
%     then $L$ can also be predictively monitored in space $O(1)$ and time $O(n)$ under strong trace prefixes.
%     \ucomment{Should we make this more formal?}
% \end{corollary}

\subsection{Deterministic predictive monitoring}
\seclabel{det-poly-blowup}

While \thmref{nondet} illustrates that
the predictive monitoring question with strong trace prefixes becomes decidable 
(assuming the analogous problem for Mazurkiewicz traces is decidable),
we remark that the use of nondeterminism may lead to exponential blow-ups 
in time and space when translating it to a deterministic machine that can then be used
in a practical predictive testing setup.
Here, in this section, we establish that one can tactfully avoid this blow-up.
In fact, we show that only allowing a polynomial multiplicative blow-up is sufficient
to do predictive monitoring under strong prefixes starting with a
deterministic Turing machine for that works under trace equivalence.

Our result is inspired by prior works on predictive 
monitoring under trace languages~\cite{Bertoni1989}.
Here, one identifies strong ideals\footnote{A strong ideal of an execution $\tr = e_1, \ldots, e_k$  is a subset of $\set{e_1, e_2, ..., e_k}$ such that for every $i < j$, if $(\lbl{e_i}, \lbl{e_j}) \in \dstrong$ and if $e_j \in X$, then $e_i \in X$}, i.e., sets of events that are
downward closed with respect to the \emph{strong} dependence relation,
and checks whether there is a linearization of one of them 
that respects both strong and weak dependence and also belongs to the target language $L$.
A parameter that crucially determines the time complexity here
is the \emph{width} of the concurrency alphabet.
In our setting, the width $\alpha_\dstrong$ is the size of the largest subset of $\alphabet$,
that contains no two letters which are dependent according to $\dstrong$.

\begin{restatable}{theorem}{detalpha}
    \thmlabel{detalpha} 
    Fix a language $L$. 
    Let $M$ be a deterministic Turing machine that uses time $T(n)$ and space $S(n)$, 
    such that $L(M) = \setpred{w}{\eqcl{\dep}{w}\cap L \neq \emptyset}$.
    Then, there exists a deterministic Turing machine $M'$ 
    that runs in time
    % $T(n) \cdot n^{\alpha_\dstrong}$ \ucomment{Perhaps $O((n+T(n)) \cdot n^{\alpha_\dstrong})$}
    $O((n+T(n)) \cdot n^{\alpha_\dstrong})$
    and uses space 
    % $S(n) \cdot n^{\alpha_\dstrong}$ \ucomment{Shouldn't this be $S(n) + O(n)$}, 
    $S(n) + O(n)$,
    such that $L(M') = \setpred{w}{\strongdcl{w}{\dstrong}{\dweak} \cap L \neq \emptyset}$.
\end{restatable}

The above complexity bounds follow because
one can systematically enumerate
those subsequences of the input $w$ 
which are downward closed with respect to $\dstrong$,
by in turn enumerating the space of strong ideals.
The set of strong ideals is, in turn, bounded by 
$n^{\alpha_\dstrong}$ where $n$ is the length of the input string $w$. 

\subsection{Ensuring well-formedness  and soundness}
\seclabel{well-formedness}

Recall that the dual concurrent alphabet $(\alphabet_\rwl, \dstrong_\rwl, \dweak_\rwl)$
is not sufficient by itself for ensuring that the strong trace prefixes of an
execution $\tr \in L_\wf$ are also well-formed.
Well-formedness can nevertheless be retrofitted in the predictive monitoring
algorithm with same additional time, space and non-determinism.
\thmref{well-formedness} formalizes this and follows from the observation that
the set $L_\wf$ is (a) regular, and 
(b) closed under strong prefixes, i.e.,
for every $\tr \in L_\wf$, we have $\strongdcl{\tr}{\dstrong_\rwl}{\dweak_\rwl} \subseteq L_\wf$,
and algorithms for predictive monitoring can be easily augmented
to reason about the set 
$\strongdcl{\tr}{\dstrong_\rwl}{\dweak_\rwl}\cap L_\wf$.

\begin{restatable}{theorem}{wellformedness}
\thmlabel{well-formedness}
    Let $L \subseteq \alphabet_\rwl^*$ 
    and let $M$ be a deterministic Turing machine, that uses time $T(n)$ and space $S(n)$, 
    such that $L(M) = \setpred{\tr \in L_\wf}{\eqcl{\dep_\rwl}{\tr} \cap L_\wf \cap L \neq \emptyset}$.
    There is a nondeterministic Turing machine $M'$ (resp. deterministic Turing machine $M''$) that uses time 
    $T(n) + O(n)$ (resp. $O((n+T(n)) \cdot n^{\alpha_\dstrong})$)
    and space $S(n) + O(n)$ 
    such that $L(M') ( = L(M'')) = \setpred{\tr \in L_\wf }{\strongdcl{\tr}{\dstrong}{\dweak} \cap L_\wf \cap L \neq \emptyset}$.
    Moreover, if $M$ runs in one-pass,
    $M'$ uses space $S(n) + c$ (for some constant $c$).
\end{restatable}

%% file: readsfrom.tex
%!TEX root=./main.tex

\section{Strong Reads-from Prefixes}

\seclabel{rf}

Strong trace prefixes generalize Mazurkiewicz traces and
can enhance precision of predictive monitoring algorithms. 
In this section, we propose further generalizations
in the context of $\alphabet_\rwl$, bringing the power of 
commutativity-based reasoning
further close to correct reorderings. 
Towards this, we observe that the key constraints that correct reorderings
must satisfy are only thread-order and reads-from relation, and
thus $\dstrong_\rwl$ may be relaxed further 
by removing the dependence between writes and reads.

\input{figures/trace-limitation-rwwr}

Consider the trace $\tr$ in \figref{trace-limitation-rwwr}. 
Here, the only strong prefixes of $\tr$ 
(under $\dstrong_\rwl$ and $\dweak_\rwl$)
are its own (string) prefixes.
That is, strong trace prefixes cannot be used to argue that there is 
a reordering (namely $\rho'_4$ in \figref{rwwr-race-correct-reordering}) 
in which $\wt(y)$ and $\rd(y)$ are next to each other. 
% In fact, Mazurkiewicz style commutativity reasoning can infer together with prefix reasoning. 
Intuitively, one can first obtain the 
intermediate $\rho_4$ (\figref{rwwr-no-race-trace}) from $\tr_4$
by dropping all events in the block of events containing 
$e_2$ labelled $\ev{t_1}{\wt(x)}$ together with all its read 
events $\rf{\tr_4}{e_2} = \set{e_3, e_4}$,
and then obtain $\rho'_4$ from $\rho_4$ using Mazurkiewicz between
independent events.
We remark however that, neither $\rho_4$ nor $\rho'_3$ is a strong trace prefix
of $\tr_4$ because $(\ev{t_3}{\rd(x)}, \ev{t_2}{\wt(x)}) \in \dstrong_\rwl$.
However, observe that one cannot obtain this reordering $\rho_1$ in the presence of $\dstrong_\rwl$. 
% Moreover, $\rho_1$ is Mazurkiewicz equivalent to trace $\rho_2$ which belongs to the data race language.

The above example illustrates the possibility of relaxing $\dstrong_\rwl$ 
by removing the dependencies between reads and writes. 
However, an incautious relaxation 
(such as removing $(\ev{t_1}{\wt(x)}, \ev{t_3}{\rd(x)})$ from $\dstrong_\rwl$) 
may result into a prefix like 
$\rho''_4 = \ev{t_1}{\wt(y)}\ev{t_3}{\rd(x)}\ev{t_2}{\wt(x)}\ev{t_2}{\rd(x)}\ev{t_2}{\rd(y)}$
which is not a correct reordering of $\tr_4$.
In other words, while $(\rd, \wt)$ and $(\wt, \rd)$ dependencies can be relaxed, 
the stronger semantic dependence due to reads-from must still be retained. 
As a reminder, such a relaxation cannot accurately be modelled under 
strong prefixes alone since $(\alphabet_\rwl, \dstrong_\rwl, \dweak_\rwl)$
is already the weakest alphabet (\thmref{maximality}).
We instead model this as \emph{strong reads-from prefixes} defined below:
\begin{definition}[Strong Reads-from Prefix]
    \deflabel{strong-reads-from-prefix}
    The \emph{\underline{strong reads-from prefix order}} induced by $(\alphabet_\rwl, \dstrong_\rwl, \dweak_\rwl)$,
    denoted $\rfqo$, is the smallest 
    reflexive and transitive binary relation on $\alphabet_\rwl^*$ that satisfies:
    \begin{enumerate}
        \item $\strongqo{\dstrong_\rwl}{\dweak_\rwl} \; \subseteq \; \rfqo$, and
        \item let $\tr = \tr_1 \cdot e \cdot \tr_2$,
        if $\forall e' \in \tr_2$, we have $(e, e')\not\in\dep_\threads$ and $e'\not\in\rf{\tr}{e}$, then
        $\tr_1\cdot\tr_2 \rfqo \tr_1 \cdot e \cdot \tr_2$.
    \end{enumerate}
    We say $w' \in \alphabet_\rwl^*$ is a \emph{\underline{strong reads-from prefix}} of $w$
    if $w' \rfqo w$.
    We use $\rfdcl{w} = \setpred{w'\in\alphabet_\rwl^*}{w' \rfqo w}$ to denote the \emph{\underline{strong reads-from downward closure}} of $w$.
\end{definition}
In the above example, $\rho_4$ and $\rho'_4$ now can be modelled as strong reads-from prefixes of $\tr_4$, i.e., $\rho_4, \rho'_4\in \rfdcl{\tr_4}$,
since $\ev{t_1}{w(x)}$ and $\ev{t_1}{r(x)}$ are not strong dependent with and not in the reads-from relation with any subsequent events.
The soundness and precision power of strong-reads from prefixes are clear:
\begin{restatable}{theorem}{soundnessRF}{\it\bf(Soundness and Precision Power)}.
\thmlabel{soundness}
    For each well-formed execution $\tr\in L_\wf$, we have: 
    \[
    \eqcl{\dstrong_\rwl\cup \dweak_\rwl}{\tr} \subseteq \idealdcl{\tr}{\dstrong_\rwl\cup \dweak_\rwl} \subseteq \strongdcl{\tr}{\dstrong_\rwl}{\dweak_\rwl} \cap L_\wf 
    \subseteq \rfdcl{\tr}\cap L_\wf
    \subseteq \creorder{\tr}.
    \]
Moreover, there is a $\tr\in L_\wf$ for which each of the subset relationships are strict.
\end{restatable}

We now discuss the algorithmic impact of this 
further relaxation to strong reads-from prefixes.
Here we obtain results which are analogue to \thmref{nondet} and \thmref{detalpha}.
\begin{restatable}{theorem}{readsfrom}
\thmlabel{reads-from-nondet}
    Let $L \subseteq \alphabet^*$ 
    and let $M$ be a deterministic Turing machine, that uses time $T(n)$ and space $S(n)$, 
    such that $L(M) = \setpred{w}{\eqcl{\dep}{w}\cap L \neq \emptyset}$.
    There is a nondeterministic Turing machine $M'$ (resp. deterministic Turing machine) that uses time 
    $T(n) + O(n)$ (resp. $O((n+T(n)) \cdot n^{\alpha_{\dstrong_\rwl}})$)
    and space $S(n) + O(n)$ 
    such that $L(M')  (= L(M'')) = \setpred{w}{\rfdcl{w} \cap L \neq \emptyset}$.
    Moreover, if $M$ runs in one-pass,
    then $M'$ uses space $S(n) + c$ (for some constant $c$).
\end{restatable}
% \begin{theorem}
% \thmlabel{reads-from-det}
%     Fix a language $L$. 
%     Let $M$ be a deterministic Turing machine that uses time $T(n)$ and space $S(n)$, 
%     such that $L(M) = \setpred{w}{\eqcl{\dep}{w}\cap L \neq \emptyset}$.
%     Then, there exists a deterministic Turing machine $M'$ 
%     that runs in time
%     % $T(n) \cdot n^{\alpha_\dstrong}$ \ucomment{Perhaps $O((n+T(n)) \cdot n^{\alpha_\dstrong})$}
%     $O((n+T(n)) \cdot n^{\alpha_\dstrong})$
%     and uses space 
%     % $S(n) \cdot n^{\alpha_\dstrong}$ \ucomment{Shouldn't this be $S(n) + O(n)$}, 
%     $S(n) + O(n)$,
%     such that $L(M') = \setpred{w}{\rfdcl{w}\cap L \neq \emptyset}$.
% \end{theorem}
These follow because one can guess and check whether a prefix preserves thread order and reads-from relations.
In next section, we will show that such a relaxation allows us to obtain a previously known class of data races and deadlocks.

%% file: figures/trace-limitation-rwwr.tex
%!TEX root=../main.tex

\begin{figure}[t]
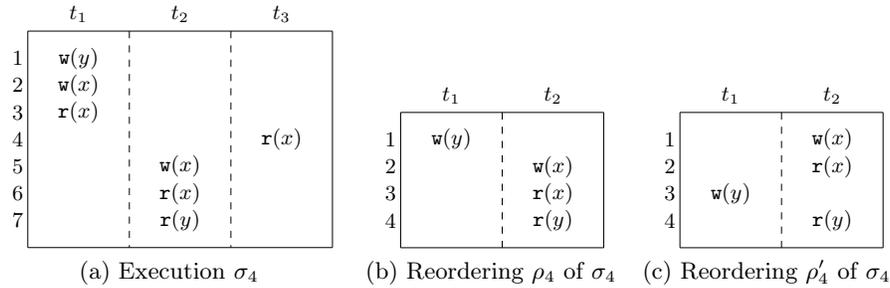

\centering
\subfloat[Execution $\tr_4$ \figlabel{rwwr-race-trace}]{%
    \execution{3}{
        \figev{1}{$\wt(y)$}
        \figev{1}{$\wt(x)$}
        \figev{1}{$\rd(x)$}
        \figev{3}{$\rd(x)$} 
        \figev{2}{$\wt(x)$}
        \figev{2}{$\rd(x)$}
        \figev{2}{$\rd(y)$}
    }
}
\hfill
\subfloat[Reordering $\rho_4$ of $\tr_4$ \figlabel{rwwr-no-race-trace}]{%
    \execution{2}{
        \figev{1}{$\wt(y)$}
        \figev{2}{$\wt(x)$}
        \figev{2}{$\rd(x)$}
        \figev{2}{$\rd(y)$}
    }
}
\hfill
\subfloat[Reordering $\rho'_4$ of $\tr_4$ \figlabel{rwwr-race-correct-reordering}]{
    \execution{2}{
        \figev{2}{$\wt(x)$}
        \figev{2}{$\rd(x)$}
        \figev{1}{$\wt(y)$}
        \figev{2}{$\rd(y)$}
    }
}
% \hfill
% \subfloat[Reordering $\rho_2$ of $\tr_2$ \figlabel{rwwr-race-correct-reordering}]{
%     \execution{3}{
%         \figev{1}{$\wt(y)$}
%         \figev{3}{$\rd(x)$} 
%         \figev{2}{$\wt(x)$}
%         \figev{2}{$\rd(x)$}
%         \figev{2}{$\rd(y)$}
%     }
% }
\begin{comment}
\subfloat[Execution $\tr_1$ \figlabel{rwwr-race-trace}]{%
    \smallexecution{3}{
        \smallfigev{1}{$\wt(y)$}
        \smallfigev{1}{$\wt(x)$}
        \smallfigev{1}{$\rd(x)$}
        \smallfigev{3}{$\rd(x)$} 
        \smallfigev{2}{$\wt(x)$}
        \smallfigev{2}{$\rd(x)$}
        \smallfigev{2}{$\rd(y)$}
    }
}
\hfill
\subfloat[Execution $\tr_2$ \figlabel{rwwr-no-race-trace}]{%
    \smallexecution{3}{
        \smallfigev{1}{$\wt(y)$}
        \smallfigev{2}{$\wt(x)$}
        \smallfigev{2}{$\rd(x)$}
        \smallfigev{2}{$\rd(y)$}
    }
}
\hfill
\subfloat[Reordering $\rho_2$ of $\tr_2$ \figlabel{rwwr-race-correct-reordering}]{
    \smallexecution{3}{
        \smallfigev{2}{$\wt(x)$}
        \smallfigev{2}{$\rd(x)$}
        \smallfigev{1}{$\wt(y)$}
        \smallfigev{2}{$\rd(y)$}
    }
}
\hfill
\subfloat[Reordering $\rho_2$ of $\tr_2$ \figlabel{rwwr-race-correct-reordering}]{
    \smallexecution{3}{
        \smallfigev{1}{$\wt(y)$}
        \smallfigev{3}{$\rd(x)$} 
        \smallfigev{2}{$\wt(x)$}
        \smallfigev{2}{$\rd(x)$}
        \smallfigev{2}{$\rd(y)$}
    }
}
\end{comment}
\caption{Execution $\tr_1$ has a predictable data race, and can be exposed with trace equivalence.
Execution $\tr_2$ has a predictable data race (witnessed by correct reordering $\rho_2$), but cannot be
exposed by a trace prefix.}
\figlabel{trace-limitation-rwwr}
\vspace{-0.1in}
\end{figure}

%% file: conn-syncp.tex
%!TEX root=./main.tex

\section{Strong Prefixes versus Synchronization Preservation}
\seclabel{syncp}

% Trace prefix is not the first notion that captures a space of feasible reorderings between correct reorderings and Mazurkiewicz trace.
% The notion of sync-preserving correct reorderings~\cite{SyncP2021} characterizes a subset of correct reorderings where the order of synchronization primitives, such as locks, are maintained. 
% These previous works propose algorithms predicting data races and deadlocks in sync-preserving correct reorderings, both of which run in linear time and linear space.
% In this section, we first show that sync-preserving correct reordering is a strict larger space than our notion, \zacomment{name}. 
% Nevertheless, we reveal the ability of our framework to report exactly the same data races and deadlock instances as in sync-preserving correct reorderings.
% According to , our approach consumes also linear time, however, only constant space. 

Recall the execution $\tr_2$ in \figref{trace-limitation-race}, where $(e_1, e_6)$ is a data race, but cannot be detected using a happens-before style detector, i.e., using the dependence $\dep_L$. 
On the other hand, trace $\rho_2$ demonstrates that this can be captured using strong prefixes (i.e., under $\dstrong_\rwl$ and $\dweak_\rwl$).
Indeed, this is a classic example of a \emph{synchronization preserving} data race proposed in~\cite{SyncP2021} 
and characterize a large class of predictable data races that can also be detected in linear time. 
The analogous notion class of synchronization preserving deadlocks captures a large class of predictable deadlocks~\cite{tuncc2023sound}, and can be detected efficiently. 
Both these classes of bugs can be predicted by looking for \emph{synchronization preserving correct reorderings}, and in this section we investigate the relationship between them and strong trace prefixes.
%for shared-memory multithreaded programs. 

\myparagraph{Synchronization preserving reorderings, data races and deadlocks}
% Let us first recall the notion of synchronization preserving data races and deadlocks. 
An execution $\rho \in \alphabet^*_\rwl$ is a synchronization preserving correct reordering of execution $\tr \in \alphabet^*_\rwl$ if
\begin{enumerate*}[label=(\alph*)]
    \item $\rho$ is a correct reordering of $\tr$, and
    \item for each pair of acquire events $a_1 \neq a_2$ (alternatively, critical sections) of $\tr$ on the same lock $\lk$, such that both $a_1$ and $a_2$ are present in $\rho$, we have, $a_1 \trord{\rho} a_2$ iff $a_1 \trord{\tr} a_2$.
\end{enumerate*}
We use $\syncp{\tr}$ to denote the set of all synchronization-preserving correct reordering of $\tr$.
A sync(hronization)-preserving data races is a pair of conflicting events $(e_1, e_2)$, such that there is a synchronization preserving correct reordering $\rho$ 
in which $e_1$ and $e_2$ are $\tr$-enabled, i.e., $e_1, e_2 \not\in \events{\rho}$ but 
all events that are thread before $e_1$ and $e_2$ are present in $\rho$.
Likewise, A sync-preserving deadlock of length $k$ is a deadlock pattern
\footnote{A deadlock pattern of size $k$ is a sequence of acquire events $D = (e_1, \dots, e_k)$ performed by $k$ distinct threads $t_1, \dots, t_k$, acquiring $k$ distinct locks $\lk_1, \dots, \lk_k$ such that
$\lk_i \in \texttt{HeldLks}(e_{i\% k + 1})$, and further, the locks held at $e_i$ and $e_j$ are disjoint ($i \neq j$).} 
$(e_1,\dots,e_k)$, 
such that there is a synchronization preserving correct reordering $\rho$ 
in which $e_1, \dots, e_k$ are $\tr$-enabled.
This class of data races and deadlocks can be detected in linear time and space~\cite{SyncP2021,tuncc2023sound}.
% \begin{theorem}
%     Given trace $\tr$, determining the existence of sync-preserving data race (deadlock) of $\tr$ can be solved in time $O(|\tr|)$ and space $O(\tr)$.
% \end{theorem}

We observe that the algorithmic 
efficiency in predicting sync-preserving data races (resp. deadlocks)
stems from the fact that 
whenever $(e_1, e_2)$ is a data race (resp. $(e_1, \dots, e_k)$ is a deadlock), 
it can be witnessed by a reordering which is not only synchronization preserving, 
but \emph{also preserves the order of conflicting read and write events},
i.e., through a \emph{conflict preserving reordering}:
\begin{definition} (Conflict-preserving Correct Reordering)
\deflabel{confp-reordering}
    A reordering $\rho$ of an execution $\tr$ is a conflict preserving correct reordering if
    \begin{enumerate*}[label=(\alph*)]
        \item $\rho$ is the correct reordering of $\tr$,
        \item for every lock $\lk$ and for any two acquire event $a_1, a_2$ labelled $\acq(\lk)$ in $\rho$, $a_1 \trord{\rho} a_2$ iff $a_1 \trord{\tr} a_2$, and 
        \item for every two conflicting events $e_1$ and $e_2$ in $\rho$, $e_1\trord{\rho} e_2$ iff $e_1 \trord{\tr} e_2$.
    \end{enumerate*} 
\end{definition}
Here, we say $(e_1, e_2)$ is a conflicting pair of events if $(\lbl{e_1}, \lbl{e_2}) \in \dep_{\mathsf{conf}}$.
We use $\confp{\tr}$ to denote all conflict-preserving correct reordering of $\tr$.
Observe that every conflict-preserving correct reordering of $\tr$ is also a synchronization-preserving correct reordering of $\tr$.

\begin{restatable}{proposition}{confpSyncpProp}
    For any execution $\tr\in \alphabet_\rwl^*$, we have $\confp{\tr}\subseteq \syncp{\tr}$.
\end{restatable}

We now formalize our observation: 
%We identify that in fact every synchronization-preserving data race (deadlock) is also a conflict-preserving data race (deadlock).
\begin{restatable}{lemma}{syncpConfpEnabledEvents}
    \lemlabel{syncp-confp}
    Let $\tr \in \alphabet_\rwl^*$ be an execution.
    A sequence of events $(e_1, \dots, e_k)$ are $\tr$-enabled in some synchronization-preserving reordering of $\tr$
    iff they are $\tr$-enabled in some conflict-preserving reordering of $\tr$.
    Thus, sync-preserving data races and deadlocks can also be witnessed using conflict preserving reorderings.
\end{restatable}

The connection between synchronization preserving and strong reads-from prefixes is now straightforward because the class of conflict preserving races or conflict preserving deadlocks can be accurately modelled in our framework:
\begin{restatable}{lemma}{confPRF}
\lemlabel{confPRF}
    Let $\tr \in \alphabet_\rwl^*$ be an execution. 
    We have $\rfdcl{\tr}\cap L_\wf = \confp{\tr}$.
\end{restatable}

Consequently, we get algorithms for detecting sync-preserving data races and deadlocks with improved space bound and same time:
\begin{restatable}{theorem}{syncPRaceDeadlockConstSpace}
    \thmlabel{syncp-constant}
    Synchronization preserving correctly synchronization races and deadlocks can be detected in linear time and constant space.
\end{restatable}

Even though, in the context of data races and deadlocks, it suffices to look at conflict preserving correct reorderings, in general, the class of synchronization preserving reorderings is much more expressive. 
As a consequence, when one goes beyond data races and deadlocks to a slightly different class of specifications, 
the predictive monitoring question under synchronization preserving reorderings quickly becomes hard. 
In particular, we demonstrate this in the context of predicting if two events can be reordered in a certain order.
Under Mazurkiewicz's trace equivalence, this problem can be decided in linear time and constant space, and thus also for strong trace prefixes (\thmref{nondet}).
However, in the context of synchronization preserving reorderings, we show that this problem cannot be solved in linear time and constant space.

\begin{restatable}{theorem}{SyncPLinearSpace}
    \thmlabel{linearspace}
    Let $\tr\in\alphabet_\rwl^*$ be an execution, and $e_1, e_2\in \events{\tr}$ be two events.
    Any streaming algorithm that checks if there is an execution $\rho\in\syncp{\tr}$ such that $e_1\trord{\rho}e_2$ uses linear space.
\end{restatable}

Indeed, the above problem (checking if two events can be flipped) is an example of a pattern language~\cite{ang2023predictive}, whose predictive monitoring can be solved in linear time and constant space under Mazurkiewicz traces, and thus also under strong prefixes.
% However, \thmref{linearspace} indicates that any streaming algorithm deciding this problem under synchronization-preserving reorderings has a linear-space lower bound.
We therefore remark that, strong prefixes lie at the horizon of tractability in the context of predictive monitoring. 

%% file: exp.tex
%!TEX root=./main.tex

\section{Experimental Evaluation}

We evaluate the effectiveness of strong prefixes and strong reads-from prefixes for the purpose of predictive monitoring of executions (over $\alphabet_\rwl^*$) of shared memory multi-threaded programs.
The goal of our evaluation is two-folds.
First, we want to empirically 
gauge the enhanced predictive power of strong reads-from prefixes
over prediction based on trace equivalence.
We demonstrate this using prediction against pattern languages proposed in~\cite{ang2023predictive}.
For data races, strong reads-from prefixes can capture sync-preserving races, 
which have already been shown to have more empirical predictive power over trace-based reasoning~\cite{SyncP2021}.
Second, we want to evaluate how our, not-so-customized but constant space, algorithm
for synchronization-preserving data races and deadlocks (\thmref{syncp-constant})
performs against the linear space algorithm due to~\cite{SyncP2021,tuncc2023sound}.
% We show that
% data races and deadlocks detected under strong reads-from prefixes is exactly the same as those are detected by the existing algorithms SyncP and SyncPD, under synchronization-preserving reorderings.

\myparagraph{Implementation and Setup}{
We implemented our predictive monitoring algorithms
for data races, deadlocks and pattern languages,
% in a dynamic analysis framework RAPID~\cite{rapid},
in Java, obtained by the determinizing the
non-deterministic monitors obtained from \thmref{reads-from-nondet}.
We evaluate against benchmarks derived from~\cite{ang2023predictive,SyncP2021,tuncc2023sound},
consisting of concurrent programs from a variety of suites:
\begin{enumerate*}[label=(\roman*)]
    \item the IBM Contest suite~\cite{FarchiNU03},
    \item the DaCapo suite~\cite{Blackburn06},
    \item the Java Grande suite~\cite{Smith01},
    \item the Software Infrastructure Repository suite~\cite{DoER05}, and
    \item others~\cite{Legunsen2016,cai2021sound,joshi2009randomized,Kalhauge2018,JulaTZC08}.
\end{enumerate*}
For each benchmark program, we generated an execution log using RV-Predict~\cite{RV2003} and 
evaluate all competing algorithms on the same execution.
Our experiments were conducted on a 64-bit Linux machine with Java 19 and 400GB heap space.
Throughout, we set a timeout of $3$ hours for each run of every algorithm.
We present brief summary of our results here, and the details can be found
in \appref{exp}.
}
\vspace{-0.15in}
\input{exp-pattern}
\vspace{-0.1in}
\input{exp-syncp}

%% file: exp-pattern.tex
%!TEX root=./main.tex

\input{tables/table-pattern}

\subsection{Enhanced Predictive Power of Strong Prefixes}

We demonstrate the enhanced predictive power due to our proposed formalism 
in the context of predictive monitoring against pattern languages~\cite{ang2023predictive}.
Pattern language specifications take the form
$\patt{a}{d}$, and thus include all executions that contain $\seq{a}{d}$
as a sub-sequence.
Predictive monitoring against pattern languages can be performed in constant space
and linear time under trace equivalence~\cite{ang2023predictive}.
% and thus, also under strong reads-from prefixes (\thmref{reads-from-nondet}).
% , i.e., a pattern can be detected in the strong reads-from prefixes downward closure but cannot be exposed under Mazurkiewicz equivalence

\myparagraph{Implementation and Methodology}{
% Our implementation conforms to \thmref{nondet}.
$M$ (under trace equivalence) proposed in~\cite{ang2023predictive}.
To perform predictive monitoring under strong reads-from prefixes,
our algorithm $M'$ guesses an appropriate prefix and invokes 
the predictive monitor $M$ (under trace equivalence) proposed in~\cite{ang2023predictive}, .
Since $M$ consumes constant space, in theory,
simulating $M'$ also requires constant space (see also \thmref{reads-from-nondet}).
The resulting space complexity, however, can be prohibitive in practice.
For scalability, we employ randomization to select a subset of prefixes, and only inspect these. 
Our results show that, despite this compromise,
the predictive power under strong reads-from prefixes is higher than under
trace equivalence.
% }
% 
% \myparagraph{Selecting Patterns}{
% \noindent
We use $30$ benchmark executions, and for each execution, we isolate $20$ patterns
(of size $3$ and $5$), from randomly chosen sub-executions of length $5000$,
following~\cite{ang2023predictive}.
For each pair of benchmark and pattern,
we run the two streaming algorithms  (trace equivalence v/s strong reads-from prefixes),
on the sub-execution from which the pattern is extracted,
allowing us to optimize memory usage.
Both algorithms terminate as soon as the pattern is matched, 
otherwise the algorithms process the entire sub-execution.
We use the publicly available implementation of~\cite{ang2023predictive}.
% (not match). 
% We also report the length of the processed prefix of the execution.
}

\myparagraph{Evaluation Results}{
Our results are summarized in~\tabref{pattern}.
First, all matches reported under trace equivalence
were also reported under strong reads-from prefixes,
as expected based on \thmref{soundness} .
Second, out of the $30\times 20$ combinations of executions and patterns,
trace equivalence based prediction reports $33$ fewer 
matches as compared to prediction based on 
strong reads-from prefixes ($466$ vs $499$).
The enhancement in prediction power spans patterns
of both sizes --- $15$ extra matches were found for patterns
of size $3$ and $18$ extra matches were found for patterns of size $5$.
Third, we also collect more fine-grained information ---
amongst the $466$ combinations reported by both,
$18$ were reported earlier (in a shorter execution prefix) under strong reads-from prefixes,
and thus are new violations.
Finally, the total time taken for prediction under reads-from prefixes
is higher, as expected, but only by $6\%$.
In summary, strong reads-from prefixes offer higher prediction power
in practice, with moderate additional overhead.
% effectively claims the enhanced predictive power of strong prefixes model over Mazurkiewicz traces.
}

%% file: tables/table-pattern.tex
%!TEX root=../main.tex

\begin{table}[t]
\caption{
\tablabel{pattern}
Predictive monitoring against pattern languages, grouped by pattern length.
Column 3 (Column 6) reports the total time taken under
trace equivalence (string reads-from prefixes).
Column 2 (Column 4) reports the number of successful matches
under trace equivalence (strong reads-from prefixes).
Column 5 reports the number of times prediction based on strong reads-from prefixes
reports an earlier match.}
\centering
\setlength{\tabcolsep}{0.1em}
\renewcommand{\arraystretch}{0.9}
\footnotesize
\scalebox{0.9}{
    \begin{tabular}{|c||c|c||c|c|c|}
        \hline
            1 & 2 & 3 & 4 & 5 & 6\\ \hline
            Length of & \#Match in & Time  & \#Match in & \#Earlier Match in & Time \\
            Pattern & Maz. & Maz. & Strong RF. & Strong RF. & Strong RF.\\\hline
            3 & 230 & 10h 23m & 249 & 7 & 10h 55m \\ 
        5 & 197 & 9h 6m & 209 & 12 & 10h 17m \\ \hline
        Total & 427 & 19h 29m & 458 & 19 & 21h 12m \\ \hline
        \end{tabular}
}
\vspace{-0.2in}
\end{table}

%% file: exp-syncp.tex
%!TEX root=./main.tex

\subsection{Strong Reads-From Prefixes v/s Sync-Preservation}

\input{tables/table-race}

We implemented the constant space linear time algorithm 
for sync-preserving data races and deadlocks (\thmref{syncp-constant})
and compare it against the linear space algorithms due to~\cite{SyncP2021,tuncc2023sound},
solving the same problem.

\myparagraph{Implementation and Methodology}{
The algorithm to detect conflict-preserving data races and deadlocks 
guesses strong reads-from prefixes and checks whether the enabled events in them constitute a
data race or a deadlock.
Following~\cite{SyncP2021,tuncc2023sound},
we filtered out thread-local events to reduce the space usage of all algorithms.

%\myparagraph{Evaluation Methods}
\noindent
% We compare our predictive monitoring algorithm under strong reads-from prefixes (conflict-preserving reorderings) with SyncP~\cite{SyncP2021} and SyncPD~\cite{tuncc2023sound} under synchronization-preserving prefixes.
We use the publicly available implementation of~\cite{tuncc2023sound,SyncP2021}.
We run all algorithms on the entire executions.
For the case of data races, we report 
the number of events $e_2$ for which there is an earlier event $e_1$ such that $(e_1, e_2)$ is a 
sync-preserving data race.
For the case of deadlocks, we report the number of
tuples of program locations corresponding to events reported to be in deadlock.
}

% \ucomment{Table captions need to be more elaborate. For example, what is N? }

\myparagraph{Evaluation Results}
We present our results in \tabref{race} and \tabref{deadlock}.
First, observe that the precision of 
data race and deadlock prediction based on strong reads-from prefixes 
is exactly the same as with prediction based on sync-preservation
(compare columns $4$ and $6$ in both tables).
Next, we observe that our implementations (even though constant-space)
is slower than the optimized algorithms proposed~\cite{SyncP2021}.
We conjecture this is because the constants appearing after determinization, 
are large (of the order of $O(2^{\textsf{poly}(|\vars| + |\threads| + |\locks|)}$),
also resulting in out-of-memory exceptions on some large benchmarks.

\input{tables/table-deadlock}

%% file: tables/table-race.tex
%!TEX root=../main.tex

\begin{table}[t]
    \caption{
    \tablabel{race}
    Synchronization-preserving v/s conflict-preserving data races.
    $\mathcal{N}$ and $\threads$ denote the number of events and threads in the executions.}
    \setlength{\tabcolsep}{0.5em}
    \renewcommand{\arraystretch}{0.9}
    \centering
    \footnotesize
    \scalebox{0.9}{
        \begin{tabular}{|c|c|c||c|c||c|c|}
            \hline
                1 & 2 & 3 & 4 & 5  & 6 & 7 \\ \hline
                Benchmark & $\mathcal{N}$ & $\threads$ & \multicolumn{2}{c||}{$\texttt{SyncP}$}  & \multicolumn{2}{c|}{$\texttt{ConfP}$} \\ \hline
                & & & Race & Time & Race & Time \\ \hline
                lang & 6K & 8 & 400 & 0.25s & 400 & 0.96s \\ 
                readerswriters & 11K & 6 & 199 & 0.70s & 199 & 1.10s \\ 
                raytracer & 15K & 4 & 8 & 0.16s & 8 & 0.27s \\ 
                ftpserver & 49K & 12 & 85 & 5.70s & 85 & 7984s \\ 
                moldyn & 200K & 4 & 103 & 0.65s & 103 & 3.24s \\ 
                derby & 1M & 5 & 29 & 12.99s & 29 & 29.3s \\ 
                jigsaw & 3M & 12 & 6 & 1.16s & 6 & 2.76s \\ 
                xalan & 122M & 9 & 37 & 149.39s & 37 & 1075 \\ 
                lufact & 134M & 5 & 21951 & 45.80s & 21951 & 29.78s \\ 
                batik & 157M & 7 & 10 & 0.11s & 10 & 0.37s \\ 
                lusearch & 217M & 8 & 232 & 7.57s & 232 & 2685s \\
                tsp & 307M & 10 & 143 & 115.8s & 143 & 631.8s \\ 
                luindex & 397M & 3 & 15 & 0.73s & 15 & 0.70s \\ \hline
                Mean & 33M & ~ & ~ & 2.82s & ~ & >22.82s \\ \hline
            \end{tabular}
    }
    \vspace{-0.2in}
\end{table}

%% file: tables/table-deadlock.tex
%!TEX root=../main.tex

\begin{table}[t]
\caption{
\tablabel{deadlock}
Synchronization-preserving v/s conflict-preserving deadlocks.
    $\mathcal{N}$ and $\threads$ denote the number of events and threads in the executions.}
\centering
\setlength{\tabcolsep}{0.5em}
\renewcommand{\arraystretch}{0.9}
\footnotesize
\scalebox{0.9}{
    \begin{tabular}{|c|c|c||c|c||c|c|}
        \hline
            1 & 2 & 3 & 4 & 5  & 6 & 7 \\ \hline
            Benchmark & $\mathcal{N}$ & $\threads$ & \multicolumn{2}{c||}{$\texttt{SyncP}$}  & \multicolumn{2}{c|}{$\texttt{ConfP}$} \\ \hline
            & & & Deadlock & Time & Deadlock & Time \\ \hline
            JDBCMySQL-1 & 442K & 3 & 2 & 0.48s & 2 & 4.45s \\ 
            JDBCMySQL-2 & 442K & 3 & 1 & 0.43s & 1 & 0.35s \\ 
            ArrayList & 2.63M & 801 & 3 & 2.47s & - & TO \\ 
            IdentityHashMap & 2.71M & 801 & 1 & 1.95s & 1 & 3.79s \\ 
            Stack & 2.93M & 801 & 3 & 3.77s & 3 & 261s \\ 
            LinkedList & 3.40M & 801 & 3 & 3.02s & 3 & 1174s \\ 
            HashMap & 3.43M & 801 & 2 & 2.34s & - & TO \\ 
            WeakHashMap & 3.48M & 801 & 2 & 2.60s & - & TO \\ 
            Vector & 3.80M & 3 & 1 & 2.58s & 1 & 7.59s \\ 
            LinkedHashMap & 4.20M & 801 & 2 & 2.42s & - & TO \\ 
            TreeMap & 9.03M & 801 & 2 & 2.44s & 2 & 1564s \\ \hline \hline
            Mean & 2.53M & ~ & ~ & 1.88s & ~ & >230.84s \\ \hline
        \end{tabular}
}
\vspace{-0.2in}
\end{table}

%% file: conclusion.tex
%!TEX root = main.tex
\vspace{-0.1in}
\section{Related Work and Conclusions}

Our work is inspired from runtime predictive analysis for testing concurrent programs,
where the task is to enhance the coverage of naive dynamic analysis techniques to a larger space of correct reorderings~\cite{MaximalCausalModel2013,Said11,RVPredict2014}.
A key focus here is to improve the scalability of prediction techniques for concurrency bugs such as data races~\cite{cp2012,wcp2017,SHB2018,Pavlogiannis2019,OSR2024,SyncP2021}, deadlocks~\cite{Kalhauge2018,tuncc2023sound}, atomicity violations~\cite{mathur2020atomicity,biswas2014doublechecker,farzan2008monitoring} and more general properties~\cite{huang2015gpredict6,ang2023predictive} for an otherwise intractable problem~\cite{Mathur2020}.
The theme of our work is to develop efficient algorithms 
for predictive concurrency bug detection.
We start with the setting of trace theory~\cite{Mazurkiewicz1987}, where
questions such as checking whether two events can be flipped,
which are intractable in general~\cite{grain2024,Mathur2020,kulkarni2021dynamic}, can be answered 
in constant space.
% More generally, the class of star-connected languages~\cite{Ochmanski85} precisely characterizes those languages that can be monitored efficiently under trace equivalence. 
The problem of relaxing Mazurkiewicz traces has been studied~\cite{Diekert94}.
Recent work~\cite{grain2024} has focused on reasoning about commutativity of \emph{grains} of events. 

In contrast, our work takes an orthogonal angle and proposes that, 
for co-safety properties, one might be able to 
perform relaxations by 
exploiting semantic properties of programming constructs in multithreaded shared-memory programs. 
To this end, strong trace prefixes and its extension on a concrete alphabet $\alphabet_\rwl$, strong reads-from prefixes,
where the commutativity between events can be stratified into strong and weak dependencies.
This simple relaxation allows us to capture a larger class of concurrency bugs, 
while still retaining the algorithmic simplicity that event based commutativity offers.
% Predictive monitoring under strong prefixes provides enhanced predictive power,
% by relaxing conservative dependence in Mazurkiewicz equivalence,
% while consumes similar algorithmic complexity compared with the problem under Mazurkiewicz equivalence,
% due to its simple framework which can be systematically monitored.
We also show connections between prior algorithms for (sync-preserving) 
data race and deadlock prediction and
our formalism, and arrive at asymptotically faster algorithms for them.
% We identify that our strong prefixes is an important notion,
% which balances between predictive power and algorithmic tractability.
We envision that combining commutativity based on
groups of events~\cite{grain2024} and prefix based prediction
may be an interesting avenue for future research.

%% file: app-proofs-strong-prefixes.tex
%!TEX root=./main.tex

\section{Proofs from \secref{def-prefix}}

\soundnessStrongPrefix*

\begin{proof}[Sketch]
Our proofs rely on the following equivalent, but inductive definitions of the corresponding closures: 

\myparagraph{Trace Equivalence}{
$\eqcl{\tr}{\dstrong_\rwl\cup\dweak_\rwl}$ is the fixpoint of the following sequence of relations:
\begin{align*}
\begin{array}{l}
T_0 = \set{\tr} \\
T_{i+1} = T_i \cup \setpred{w_1\cdot b \cdot a \cdot w_2}{w_1\cdot a\cdot b \cdot w2 \in T_i, (a, b) \not\in \dstrong_\rwl\cup\dweak_\rwl}
\end{array}
\end{align*}
}

\myparagraph{Trace Ideal}{
$\idealdcl{\tr}{\dstrong_\rwl\cup\dweak_\rwl}$ is the fixpoint of the following sequence of relations:
\begin{align*}
\begin{array}{l}
I_0 = \set{\tr} \\
I_{2i+1} = I_{2i} \cup \setpred{w_1\cdot b \cdot a \cdot w_2}{w_1\cdot a\cdot b \cdot w2 \in I_i, (a, b) \not\in \dstrong_\rwl\cup\dweak_\rwl} \\
I_{2i+2} = I_{2i+1} \cup \setpred{w_1}{\exists a \in \alphabet_\rwl, w_1 \cdot a \in I_{2i+1}}
\end{array}
\end{align*}
}

\myparagraph{Strong Prefix Ideal}{
$\strongdcl{\tr}{\dstrong_\rwl}{\dweak_\rwl}$ is the fixpoint of the following sequence of relations:
\begin{align*}
\begin{array}{l}
S_0 = \set{\tr} \\
S_{2i+1} = S_{2i} \cup \setpred{w_1\cdot b \cdot a \cdot w_2}{w_1\cdot a\cdot b \cdot w2 \in S_{2i}, (a, b) \not\in \dstrong_\rwl\cup\dweak_\rwl} \\
S_{2i+2} = S_{2i+1} \cup \setpred{w_1\cdot w_2 \in L_\wf}{\exists a \in \alphabet_\rwl, w_1 \cdot a \cdot w_2 \in S_{2i+1}, a \text{ is not dependent with any event in } w_2}
\end{array}
\end{align*}
}

Now, soundness of the above three closures is straightforward, since each
step in the fixpoint preserves the thread-order, reads-from and well-formedness.
Also, the subset relations are also obvious since trace equivalence one has the least rules,
followed by trace ideal and followed by strong prefix ideal.
\end{proof}

\maximality*

\begin{proof}[Sketch]
Follows straightforwardly from \defref{strong-trace-prefix} and \equref{rwl}
\end{proof}

%% file: app-proofs-complexity.tex
%!TEX root=./main.tex

\section{Proofs from \secref{complexity}}

\nondeterminism*

\begin{proof}[Sketch]
    We propose the construction of $M'$ here.
    In general, $M'$ first perform a nondeterministic streaming process over input $\tr=a_1\dots a_n$ to generate $\tr'$ such that $\tr'\strongqo{\dstrong}{\dweak} \tr$,
    then simulate $M$ on $\tr'$.

    The strong trace prefix generation algorithm will read the input tape in one pass and write down $\tr'$ on a scratch tape.
    
    \begin{quote}
        Strong trace prefix generation $G(\tr)$:
        \begin{enumerate}
            \item $\tr' = \epsilon$
            \item $\texttt{Drop} = \emptyset$,
            \item For each $a_i\in \tr$:
            \item \quad If $a_i \in \texttt{Drop}$:
            \item \quad \quad $\texttt{Drop} = \texttt{Drop}\cup a_i$
            \item \quad else nondetermistically execute either of the following two:
            \begin{enumerate}[label=(\alph*)]
                \item $\tr' = \tr'\cdot a_i$
                \item $\texttt{Drop} = \texttt{Drop}\cup a_i$
            \end{enumerate}
            \item return $\tr'$.
        \end{enumerate}
    \end{quote}

    Now we give the construction of $M'$, which reads $\tr'$ generated by $G(\tr)$ and simulates $M$ on $\tr'$.

    \begin{quote}
        $M'(\tr)$:
        \begin{enumerate}
            \item $\tr' = G(\tr)$
            \item return $M(\tr')$.
        \end{enumerate}
    \end{quote}

    \myparagraph{Complexity}
    Overhead of time and space complexity comes from prefix generation procedure $G(\tr)$, which takes linear time and linear space.
    Observe that the size of $\texttt{Drop}$ is bounded by the size of $\alphabet$.
    Moreover, if $M$ works in a streaming way, then $G(\tr)$ does not need to write down $\tr'$, but simply simulate $M$ on the selected $\tr'$.

    \myparagraph{Correctness}
    Let us denote all possible strings generated by $G(\tr)$ as $[G(\tr)]$.
    The correctness relies on the observation that $\eqcl{\dstrong\cup \dweak}{[G(\tr)]} = \strongdcl{\tr}{\dstrong}{\dweak}$.
    
    \noindent($\eqcl{\dstrong\cup \dweak}{[G(\tr)]} \subseteq \strongdcl{\tr}{\dstrong}{\dweak}$)
    Notice that $[G(\tr)]$ is a set of subsequence of $\tr$.
    Take any $\tr'\in[G(\tr)]$.
    For every dropped element $a_i$ from $\tr$, the latter elements which are included in $\tr'$ cannot be strong dependent with $a_i$, due to $\texttt{if}$ checking.

    \noindent($\strongdcl{\tr}{\dstrong}{\dweak} \subseteq \eqcl{\dstrong\cup \dweak}{[G(\tr)]}$) 
    Observe that for any $\tr'' \in \strongdcl{\tr}{\dstrong}{\dweak}$, it can be reordered under Mazurkiewicz equivalence as $\tr'$ such that $\tr'$ is a subsequence of $\tr$.
    Let us prove that $\tr' \in [G(\tr)]$.
    For any element $a_i$ in $\tr'$ there cannot be any earlier dropped element $a_j$ in $\tr$ such that $(a_i, a_j) \in\dstrong$, otherwise dropping $a_j$ violates the definition of strong trace prefixes.
    Thus, there is a nondeterministic option to add $a_i$ into $\tr'$.
    For any dropped element $a'_i$, $G(\tr)$ can always drop it.
\end{proof}

\detalpha*

\begin{proof}[Sketch]
    We give the construction of $M'$ here:
    \begin{quote}
        $M'(\tr)$:
        \begin{enumerate}
            \item Let $\tr = e_1e_2\dots e_n$
            \item For each tuple $\tau = (e_1, e_2, \dots, e_{\alpha_\dstrong})$:
            \item \quad Let $X$ be the $\dstrong$-downward closure of $\tau$
            \item \quad $\tr' = X$ projected to $\tr$
            \item \quad If $\tr'$ is a strong trace prefix of $\tr$:
            \item \quad \quad If $M(\tr')$ is true, return 1
            \item return 0
        \end{enumerate}
    \end{quote}

    The process of checking validity of strong trace prefix is similar to generation. 
    It ensures that the generation procedure can reach line $6(a)$ when an element is presented in $\tr'$.
    
    \myparagraph{Complexity}
    The time complexity overhead comes from enumerate all $\alpha_\dstrong$ tuples.
    We require $\alpha_\dstrong$ pointer to incrementally enumerate all tuples.
    Furthermore, we need linear space to construct $X$.
    The total space overhead is $O(\alpha_\dstrong\log n) + O(n) = O(n)$.

    \myparagraph{Correctness}
    Observe that for any set of events from $\tr$, at most $\alpha_\dstrong$ events are $\dstrong$ independent with each other.
    Thus, $M'$ enumerates all possible `linearizations' based on $\dstrong$.
    The remaining correctness follows from the previous proof.
\end{proof}

\wellformedness*

\begin{proof}[Sketch]
    The construction of $M'$ and $M''$ is similar to the proof of \thmref{nondet} and \thmref{detalpha}.
    The only modification is to add additional well-formedness checking a strong prefix $\tr'$, which is a regular property and can be checked by a DFA.
\end{proof}

%% file: app-proofs-rf.tex
%!TEX root=./main.tex

\section{Proofs from \secref{rf}}

\soundnessRF*

\begin{proof}[Sketch]
Proof follows from \thmref{soundness} and \defref{strong-reads-from-prefix}.
\end{proof}

\readsfrom*

\begin{proof}[Sketch]
    The construction of $M'$ and $M''$ is similar to the proof of \thmref{nondet} and \thmref{detalpha}.
    The only modification is on checking the validity of strong reads-from prefixes, which can be checked in a streaming fashion by maintaining, for every memory location $x$, whether the last write on $x$ has been dropped; if so, all the subsequent reads on $x$ until the next write on $x$ must also be dropped.
\end{proof}

%% file: app-proofs-syncp.tex
%!TEX root=./main.tex

\section{Proofs from \secref{syncp}}

\confpSyncpProp*

\begin{proof}[Sketch]
Follows straightforwardly from the definitions of sync-preserving reorderings and conflict-preserving reorderings (\defref{confp-reordering}) since the former also enforce that the order on conflicting events be preserved.
\end{proof}

\syncpConfpEnabledEvents*

\newcommand{\spideal}[2]{\mathsf{SyncPIdeal}_{#1}(#2)}

\begin{proof}[Sketch] 
The proof relies on the observation that the events of any sync-preserving reordering
can be arranged in the order of the original execution $\tr$, and the resulting
execution will still be a sync-preserving correct reordering.
\begin{claim}
Let $\pi$ be a sync-preserving correct reordering of $\tr$.
Then, $\rho = \proj{\tr}{\events{\pi}}$ a conflict-preserving correct reordering of $\tr$.
\end{claim}
First, observe that $\pi$ is a correct reordering
and thus $\events{\pi}$ are downward closed with respect to thread-order and reads-from.
Further, for any lock $\lk$, if an acquire on $\lk$ is in $\pi$,
then all previous acquires must also be in $\pi$.
Finally, since the order of events is the same between $\rho$ and $\tr$,
it follows that $\rho$ is well-formed, and observesthe same thread-order and reads-from as $\tr$.
This finishes the proof since preserving $\tr$ also preserves the order of conflicting events.
\end{proof}

\confPRF*

\begin{proof}[Sketch] 
Follows straightforwardly from \defref{confp-reordering} and \defref{strong-reads-from-prefix}.
\end{proof}

\syncPRaceDeadlockConstSpace*

\begin{proof}[Sketch] 
Follows \lemref{confPRF}, \thmref{reads-from-nondet} and the fact that the property `check if $e_1, \ldots, e_k$ are enabled in a prefix' can be monitored in constant space.
\end{proof}

\SyncPLinearSpace*

\begin{proof}[Sketch]
Our construction is based on the construction for the linear space hardness result of~\cite[Theorem~3.1]{grain2024}
from the parametrixed language $L_n = \setpred{\vec{a}\#\vec{b}}{\vec{a} =  \vec{b} \in \set{0,1}^n}$,
for which any machine that recognizes it must use $\Omega(n)$ space.
Instead of repeating the construction here, we refer the readers to the technical report~\cite{grain2023techreport} corresponding to~\cite{grain2024} (see pages 31 and 32).

Now, our result follows and because the construction of~\cite{grain2024} does not use any locks, 
and because in absence of locks, all sync-preserving reorderings are simply correct reorderings!
Of course, care must be taken since sync-preserving reorderings may have fewer events
than executions that are reads-from equivalent a la~\cite{grain2024} which require all events to be present.
This can be remedies by asking the ordering questions between the two read events
$\rd(u)$ at the end of each thread $t_1$ and $t_2$, for which the answer is equivalent
to that between the $\rd(u)$ of $t_1$ and $\wt(u)$ of $t_2$.
Asking the question on the last events of both threads ensures that the sync-preserving reorderings we 
are interested in also include all events of the original execution.
\end{proof}

%% file: app-exp.tex
\section{Detailed Experimental Results}
\applabel{exp}

\input{tables/table-pattern-full-3}

\input{tables/table-pattern-full-5}
\input{tables/table-race-full}
\input{tables/table-deadlock-full}

%% file: tables/table-pattern-full-3.tex
%!TEX root=../main.tex

\begin{table}[t]
    \caption{
    \tablabel{pattern-full}
    Predictive monitoring against pattern languages of pattern length 3.
    $\threads$ denote the number of threads in the executions.
    Column 4 (Column 7) reports the total time taken under
    trace equivalence (string reads-from prefixes).
    Column 3 (Column 5) reports the number of successful matches
    under trace equivalence (strong reads-from prefixes).
    Column 6 reports the number of times prediction based on strong reads-from prefixes
    reports an earlier match.}
    \centering
    \setlength{\tabcolsep}{0.1em}
    \renewcommand{\arraystretch}{1.0}
    \footnotesize
    \scalebox{0.9}{
    \begin{tabular}{|c|c||c|c|c|c|c|}
        \hline
        1 & 2 & 3 & 4 & 5 & 6 & 7 \\ \hline
        & & \multicolumn{5}{c|}{Pattern of Len. 3}\\\hline
        Benchmark & $\threads$ & \#Match in & Time  & \#Match in & \#Earlier Match in & Time \\
        & & Maz. & Maz. & Strong RF. & Strong RF. & Strong RF. \\\hline
        sor & 4 & 10 & 49m 52s & 10 & 0 & 49m 31s \\ 
        avrora & 8 & 10 & 1h 20m & 10 & 0 & 1h 26m \\ 
        metafacture & 3 & 3 & 8s & 3 & 0 & 53s \\ 
        lusearch & 4 & 9 & 1h 51m & 9 & 0 & 1h 25m \\ 
        zookeeper-1 & 10 & 8 & 4s & 9 & 1 & 25m 50s \\ 
        zookeeper-2 & 18 & 7 & 19s & 10 & 0 & 12m 30s \\
        raytracer & 4 & 10 & 1h 21m & 10 & 0 & 40m 27s \\ 
        crawler-1 & 2 & 2 & 1s & 5 & 0 & 17s \\ 
        crawler-2 & 2 & 2 & 1s & 3 & 0 & 12s \\ 
        logstash-1 & 3 & 10 & 7s & 10 & 1 & 10s \\ 
        logstash-2 & 11 & 8 & 5s & 10 & 0 & 2m 23s \\ 
        exp4j-1 & 11 & 10 & 16s & 10 & 0 & 14s \\ 
        exp4j-2 & 2 & 5 & 0s & 5 & 0 & 0s \\ 
        batik & 4 & 3 & 12m 21s & 6 & 0 & 12m 28s \\ 
        moldyn & 4 & 10 & 26m 23s & 10 & 0 & 27m 15s \\ 
        xalan & 3 & 7 & 50m 32s & 7 & 0 & 1h 23m \\ 
        series & 4 & 10 & 54m 12s & 10 & 0 & 49m 9s \\ 
        montecarlo & 4 & 9 & 58m 17s & 10 & 0 & 1h 6m \\ 
        junit4-1 & 101 & 9 & 0s & 10 & 2 & 34s \\ 
        junit4-2 & 101 & 9 & 0s & 10 & 2 & 13s \\ 
        junit4-3 & 501 & 10 & 2s & 10 & 0 & 58s \\ 
        junit4-4 & 3 & 9 & 1s & 10 & 1 & 14s \\ 
        tomcat & 11 & 5 & 23s & 5 & 0 & 33s \\ 
        cassandra-1 & 3 & 8 & 11s & 9 & 0 & 2m 36s \\ 
        cassandra-2 & 3 & 8 & 13s & 9 & 0 & 1m 12s \\ 
        cassandra-3 & 3 & 9 & 6s & 9 & 0 & 27s \\ 
        cassandra-4 & 4 & 9 & 10s & 10 & 0 & 27s \\ 
        cassandra-5 & 4 & 10 & 11s & 10 & 0 & 1m 33s \\ \hline
        ~ & ~ & 230 & 10h 23m & 249 & 7 & 10h 55m \\ \hline
    \end{tabular}
}
\vspace{-0.1in}
\end{table}

%% file: tables/table-pattern-full-5.tex
%!TEX root=../main.tex

\begin{table}[t]
    \caption{
    \tablabel{pattern-full}
    Predictive monitoring against pattern languages of pattern length 5.
    $\threads$ denote the number of threads in the executions.
    Column 4 (Column 7) reports the total time taken under
    trace equivalence (string reads-from prefixes).
    Column 3 (Column 5) reports the number of successful matches
    under trace equivalence (strong reads-from prefixes).
    Column 6 reports the number of times prediction based on strong reads-from prefixes
    reports an earlier match.}
    \centering
    \setlength{\tabcolsep}{0.1em}
    \renewcommand{\arraystretch}{1.0}
    \footnotesize
    \scalebox{0.9}{
    \begin{tabular}{|c|c||c|c|c|c|c|}
        \hline
        1 & 2 & 3 & 4 & 5 & 6 & 7 \\ \hline
        & & \multicolumn{5}{c|}{Pattern of Len. 5}\\\hline
        Benchmark & $\threads$ & \#Match in & Time  & \#Match in & \#Earlier Match in & Time \\
        & & Maz. & Maz. & Strong RF. & Strong RF. & Strong RF. \\\hline
        sor & 4 & 9 & 32m 36s & 9 & 0 & 32m 16s \\ 
        avrora & 8 & 10 & 1h 39m & 8 & 2 & 1h 41m \\ 
        metafacture & 3 & 0 & 8s & 1 & 0 & 28s \\ 
        lusearch & 4 & 8 & 1h 11m & 8 & 0 & 1h 59m \\ 
        zookeeper-1 & 10 & 5 & 4s & 5 & 1 & 16m 11s \\ 
        zookeeper-2 & 18 & 8 & 14s & 10 & 3 & 3m 44s \\ =
        raytracer & 4 & 10 & 55m 54s & 10 & 0 & 47m 23s \\ 
        crawler-1 & 2 & 0 & 1s & 1 & 0 & 16s \\ 
        crawler-2 & 2 & 0 & 1s & 1 & 0 & 12s \\ 
        logstash-1 & 3 & 10 & 8s & 10 & 2 & 9s \\ 
        logstash-2 & 11 & 2 & 8s & 3 & 0 & 13s \\ 
        exp4j-1 & 11 & 10 & 13s & 10 & 0 & 11s \\ 
        exp4j-2 & 2 & 5 & 0s & 5 & 0 & 0s \\ 
        batik & 4 & 0 & 13m 28s & 0 & 0 & 13m 3s \\ 
        moldyn & 4 & 10 & 24m 52s & 10 & 0 & 25m 2s \\ 
        xalan & 3 & 5 & 1h 7m & 5 & 0 & 1h 11m \\ 
        series & 4 & 9 & 1h 3m & 9 & 0 & 58m 17s \\ 
        montecarlo & 4 & 9 & 36m 51s & 9 & 0 & 42m 17s \\ 
        junit4-1 & 101 & 10 & 0s & 10 & 1 & 16s \\ 
        junit4-2 & 101 & 9 & 0s & 10 & 0 & 15s \\ 
        junit4-3 & 501 & 10 & 1s & 10 & 2 & 1m 23s \\ 
        junit4-4 & 3 & 6 & 1s & 10 & 1 & 10s \\ 
        tomcat & 11 & 1 & 24s & 1 & 0 & 23s \\  
        cassandra-1 & 3 & 7 & 10s & 7 & 0 & 14s \\ 
        cassandra-2 & 3 & 8 & 11s & 8 & 0 & 15s \\ 
        cassandra-3 & 3 & 8 & 7s & 10 & 0 & 49s \\ 
        cassandra-4 & 4 & 10 & 4s & 10 & 0 & 23s \\ 
        cassandra-5 & 4 & 8 & 10s & 9 & 0 & 14s \\ \hline 
        ~ & ~ & 197 & 9h 6m & 290 & 12 & 10h 17m \\ \hline
    \end{tabular}
}
\vspace{-0.1in}
\end{table}

%% file: tables/table-race-full.tex
%!TEX root=../main.tex

\begin{table}[t]
    \caption{
    \tablabel{race}
    Synchronization-preserving v/s conflict-preserving data races.
    $\mathcal{N}$ and $\threads$ denote the number of events and threads in the executions.}
    \setlength{\tabcolsep}{0.5em}
    \renewcommand{\arraystretch}{0.9}
    \centering
    \footnotesize
    \scalebox{0.9}{
        \begin{tabular}{|c|c|c||c|c||c|c|}
            \hline
                1 & 2 & 3 & 4 & 5  & 6 & 7 \\ \hline
                Benchmark & $\mathcal{N}$ & $\threads$ & \multicolumn{2}{c||}{$\texttt{SyncP}$}  & \multicolumn{2}{c|}{$\texttt{ConfP}$} \\ \hline
                & & & Race & Time & Race & Time \\ \hline
                array & 51 & 4 & 0 & 0.07s & 0 & 0.18s \\ 
        critical & 59 & 5 & 3 & 0.07s & 3 & 0.04s \\ 
        account & 134 & 5 & 3 & 0.07s & 3 & 0.04s \\ 
        airlinetickets & 140 & 5 & 8 & 0.28s & 8 & 0.15s \\ 
        pingpong & 151 & 7 & 8 & 0.08s & 8 & 0.42s \\ 
        twostage & 193 & 13 & 4 & 0.14s & 4 & 0.28s \\ 
        wronglock & 246 & 23 & 25 & 0.20s & 25 & 0.17s \\ 
        boundedbuffer & 332 & 3 & 3 & 0.07s & 3 & 0.09s \\ 
        producerconsumer & 658 & 9 & 1 & 0.14s & 1 & 0.16s \\ 
        clean & 1.0K & 10 & 60 & 0.22s & 60 & 0.50s \\ 
        mergesort & 3.0K & 6 & 3 & 0.09s & 3 & 0.47s \\ 
        bubblesort & 4.0K & 13 & 269 & 0.56s & 269 & 58.41s \\ 
        lang & 6.0K & 8 & 400 & 0.25s & 400 & 0.96s \\ 
        readerswriters & 11K & 6 & 199 & 0.70s & 199 & 1.10s \\ 
        raytracer & 15K & 4 & 8 & 0.16s & 8 & 0.27s \\ 
        bufwriter & 22K & 7 & 8 & 0.66s & 8 & 5.97s \\ 
        ftpserver & 49K & 12 & 85 & 5.70s & 85 & 7984s \\ 
        moldyn & 200K & 4 & 103 & 0.65s & 103 & 3.24s \\ 
        derby & 1.0M & 5 & 29 & 12.99s & 29 & 29.30s \\ 
        linkedlist & 1.0M & 13 & 7095 & 118.08s & - & TO \\ 
        jigsaw & 3.0M & 12 & 6 & 1.16s & 6 & 2.77s \\ 
        sunflow & 11M & 17 & 119 & 2.18s & 119 & 169s \\ 
        cryptorsa & 58M & 9 & 35 & 115.54s & 35 & 258s \\ 
        xalan & 122M & 9 & 37 & 149.39s & 37 & 1075 \\ 
        lufact & 134M & 5 & 21951 & 45.80s & 21951 & 29.78s \\ 
        batik & 157M & 7 & 10 & 0.11s & 10 & 0.37s \\ 
        lusearch & 217M & 8 & 232 & 7.57s & 232 & 2685s \\ 
        tsp & 307M & 10 & 143 & 115.83s & 143 & 631.84s \\ 
        luindex & 397M & 3 & 15 & 0.73s & 15 & 0.70s \\ 
        sor & 606M & 5 & 0 & 10.45s & 0 & 27s \\\hline 
        Mean & ~ & ~ & ~ & 1.04s & ~ & >8.44s \\ \hline
            \end{tabular}
    }
    \vspace{-0.2in}
\end{table}

%% file: tables/table-deadlock-full.tex
%!TEX root=../main.tex

\begin{table}[t]
\caption{
\tablabel{deadlock}
Synchronization-preserving v/s conflict-preserving deadlocks.
    $\mathcal{N}$ and $\threads$ denote the number of events and threads in the executions.}
\centering
\setlength{\tabcolsep}{0.5em}
\renewcommand{\arraystretch}{0.9}
\footnotesize
\scalebox{0.9}{
    \begin{tabular}{|c|c|c||c|c||c|c|}
        \hline
            1 & 2 & 3 & 4 & 5  & 6 & 7 \\ \hline
            Benchmark & $\mathcal{N}$ & $\threads$ & \multicolumn{2}{c||}{$\texttt{SyncP}$}  & \multicolumn{2}{c|}{$\texttt{ConfP}$} \\ \hline
            & & & Deadlock & Time & Deadlock & Time \\ \hline
        Picklock & 66 & 3 & 1 & 0.08s & 1 & 0.08s \\ 
        Bensalem & 68 & 4 & 1 & 0.08s & 1 & 0.08s \\ 
        Test-Dimminux & 73 & 3 & 2 & 0.09s & 2 & 0.10s \\ 
        StringBuffer & 74 & 3 & 2 & 0.09s & 2 & 0.08s \\ 
        Test-Calfuzzer & 168 & 5 & 1 & 0.08s & 1 & 0.21s \\ 
        DiningPhil & 277 & 6 & 1 & 0.09s & 1 & 0.30s \\ 
        HashTable & 318 & 3 & 2 & 0.10s & 2 & 0.20s \\ 
        Log4j2 & 1.49K & 4 & 1 & 0.35s & 1 & 0.56s \\ 
        Dbcp1 & 2.16K & 3 & 2 & 0.12s & 2 & 0.12s \\ 
        Derby2 & 2.55K & 3 & 1 & 0.09s & 1 & 0.09s \\ 
        jigsaw & 143K & 21 & 1 & 2.02s & - & TO \\ 
        JDBCMySQL-1 & 442K & 3 & 2 & 0.48s & 2 & 4.45s \\ 
        JDBCMySQL-2 & 442K & 3 & 1 & 0.43s & 1 & 0.35s \\ 
        JDBCMySQL-3 & 443K & 3 & 1 & 0.49s & 1 & 4.76s \\ 
        JDBCMySQL-4 & 443K & 3 & 2 & 1.12s & 2 & 5.14s \\ 
        ArrayList & 2.63M & 801 & 3 & 2.47s & - & TO \\ 
        IdentityHashMap & 2.71M & 801 & 1 & 1.95s & 1 & 3.79s \\ 
        Stack & 2.93M & 801 & 3 & 3.77s & 3 & 261s \\ 
        LinkedList & 3.40M & 801 & 3 & 3.02s & 3 & 1174s \\ 
        HashMap & 3.43M & 801 & 2 & 2.34s & - & TO \\ 
        WeakHashMap & 3.48M & 801 & 2 & 2.60s & - & TO \\ 
        Vector & 3.80M & 3 & 1 & 2.58s & 1 & 7.59s \\ 
        LinkedHashMap & 4.20M & 801 & 2 & 2.42s & - & TO \\ 
        TreeMap & 9.03M & 801 & 2 & 2.44s & 2 & 1564s \\ \hline
        Mean & ~ & ~ & ~ & 0.52s & ~ & >9.95s \\ \hline
        \end{tabular}
}
\vspace{-0.2in}
\end{table}